\documentclass[aps,prl,reprint,superscriptaddress,longbibliography,floatfix]{revtex4-2}
\usepackage{amsmath,amssymb,amsthm,braket,bm,graphicx,booktabs}
\usepackage[colorlinks=true,citecolor=blue,urlcolor=blue,linkcolor=blue]{hyperref}
\newcommand{\Tr}{\operatorname{Tr}}
\newcommand{\MP}{\mathrm{MP}}
\newcommand{\id}{\mathrm{id}}
\newcommand{\DQ}{D_{\mathrm Q}^{(0)}}
\hypersetup{pdftitle={Exponential Advantage of Quantum over Classical References in Leakage Detection}}
\newcommand{\FundingStatement}{}
\newcommand{\DataCodeAvailabilityStatement}{The data and code supporting the findings of this study are available from the corresponding author on request.}

\newcommand{\calA}{\mathcal A}
\newcommand{\calR}{\mathcal R}
\newcommand{\Qzero}{D_{\rm Q}^{(0)}}
\newtheorem{theorem}{Theorem}
\newtheorem{lemma}{Lemma}
\newtheorem{corollary}{Corollary}
\begin{document}
\title{Exponential Advantage of Quantum over Classical References in Leakage Detection}
\author{Zheng An}
\email{anzheng@quantumsc.cn}
\affiliation{Quantum Science Center of Guangdong-Hong Kong-Macao Greater Bay Area (Guangdong), Shenzhen 518045, China}
\begin{abstract}
Leakage from a known encoding subspace can be detected by projection. However, the corresponding projector is unavailable when the encoding subspace is not classically specified. Here we show how independent quantum references prepared by the same encoder enable leakage detection without a classical description of the encoding. A coherent measurement on the references and a single message leaves every state within the two-dimensional encoding subspace, including its entanglement with a remote system, exactly unchanged. We derive the exact detection law for $M$ ideal references and prove optimality among tests with zero false alarm for every encoding. For orthogonal leakage, the miss probability is asymptotic to $4/M$, independent of the ambient dimension $d$. Measuring all references first, even collectively, gives zero detection at every finite budget under the same zero-false-alarm requirement. For a fixed detection target between zero and one, the optimal measurement-first cost is $\Theta(d/\epsilon)$ at sufficiently small tolerance $\epsilon$ on normal false alarm and conditional disturbance; a dimension-independent coherent budget suffices. For one logical qubit encoded in three physical qubits, seven coherent references achieve at least $50\%$ orthogonal-leakage detection, whereas any measurement-first receiver needs at least $588$ under the same $1\%$ normal tolerances. Quantum references thus support leakage checks without classical reconstruction, with a sample advantage exponential in the number of physical qubits.
\end{abstract}
\maketitle

Can a receiver check whether a quantum message has left its encoding subspace without knowing that subspace or disturbing a valid message? A logical qubit occupies a two-dimensional encoding subspace. A known code subspace can be tested with its projector~\cite{ChenSubspace2025}. This projection preserves every in-subspace state and its remote entanglement. We call population outside the encoding \emph{subspace leakage}, whether or not the error also drives the system outside the devices' computational levels. Leakage from these levels provides one physical setting~\cite{AliferisTerhal2007,VarbanovEtAl2020,MarquesEtAl2023,BultinkEtAl2020}. Here the encoding projector is unknown and is supplied through independent quantum references.

Independent references reveal information about the encoding without supplying its classical description. Our ideal references are maximally mixed within the subspace and independent of one another and of the message, conditional on the encoding. An uncharacterized isometric encoder can prepare these references from independent maximally mixed logical inputs and encode a separately prepared message. The references are samples of the encoding, not copies of the message. The references can be measured to estimate the subspace, for example by tomography~\cite{HaahEtAl2017,ODonnellWright2016Tomography,AnEtAl2024UnifiedTomography}. Alternatively, they can be retained for a joint operation with the message (Fig.~\ref{fig:task}). We compare these approaches without prior calibration.

Quantum memory and collective measurements can reduce the sample cost of learning quantum properties~\cite{ChenCotlerHuangLi2021,BubeckChenLi2020,HuangKuengPreskill2021,AharonovCotlerQi2022,HuangEtAl2022}. Quantum samples can also program operations, as in density-matrix exponentiation~\cite{LloydMohseniRebentrost2014,KimmelEtAl2017}. Work on coherent inference has established dimension-dependent sample advantages for quantum-output tasks, including operations on independent data~\cite{CoherentInference2026}. Here the performance criteria are asymmetric: false alarm and disturbance are constrained on normal messages, while anomalous messages may be missed. We determine the exact zero-false-alarm detection law and the measurement-first sample cost for this task.

Building on symmetry-based state comparison and one-sided rank testing~\cite{BarnettCheflesJex2003,JexAnderssonChefles2003,CheflesAnderssonJex2004,PangWu2011,ODonnellWright2017}, we construct a coherent leakage check that preserves normal messages for every unknown two-dimensional encoding subspace. The binary check detects joint components that cannot occur when all inputs share a two-dimensional support. It reads only an accept/reject flag, leaving every normal message and its remote entanglement unchanged. We derive its detection probability and prove optimality among all checks with zero false alarm for every encoding. For orthogonal leakage, the miss probability is asymptotic to $4/M$ for $M$ references, independently of the ambient dimension $d$.

In contrast, no finite-budget measurement-first receiver can detect leakage while maintaining zero false alarm for every encoding. A posterior bound and a matching construction determine the reference cost once a nonzero tolerance is allowed. For any fixed positive detection target below one and sufficiently small $\epsilon$ bounding false alarm and conditional disturbance, the optimal measurement-first reference cost is $\Theta(d/\epsilon)$, whereas a dimension-independent coherent budget suffices. For example, for one logical qubit encoded in three physical qubits, seven coherent references achieve at least 50\% orthogonal-leakage detection with false alarm and conditional disturbance each at most 1\%; every measurement-first receiver requires at least 588. The resulting sample advantage is exponential in the number of physical qubits.

\begin{figure*}[t]
\centering\includegraphics[width=\textwidth]{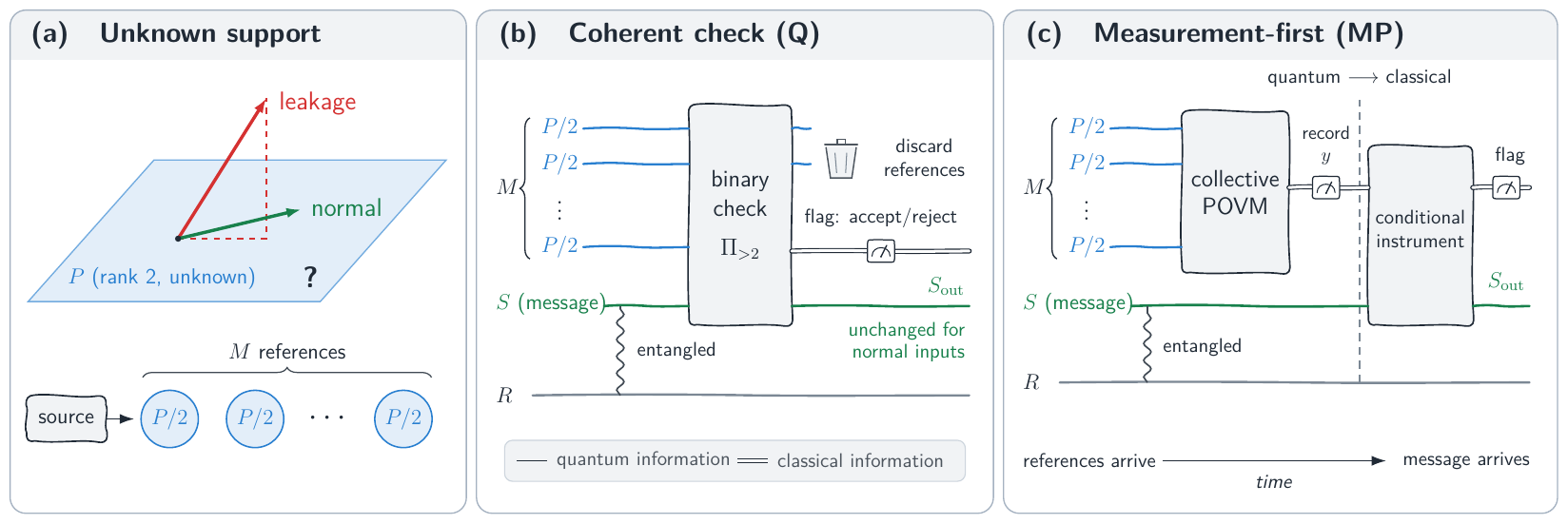}
\caption{Nondisturbing leakage detection with quantum references. (a) The unknown two-dimensional subspace onto which $P$ projects (blue plane) contains normal messages (green); leakage extends outside it (red). The source supplies $M$ references in state $P/2$, mutually independent and independent of the message conditional on $P$. (b) The quantum receiver (Q) jointly processes the references and message $S$ with a binary instrument whose rejection projector is $\Pi_{>2}$. It returns the message $S_{\mathrm{out}}$ and an accept/reject flag, exactly preserving every normal message and its entanglement with the remote system $R$; the references are discarded. (c) A measure-and-process receiver (MP) performs a collective POVM on the references and retains a classical record $y$, which controls a later instrument returning $S_{\mathrm{out}}$ and a flag. The dashed boundary marks the conversion of reference information to a classical record. In both circuits, $R$ is untouched. Single and double lines denote quantum and classical information, respectively; meter symbols denote classical readout.}\label{fig:task}
\end{figure*}

\emph{Task and performance criteria.---}Let $P$ be an unknown rank-two projector on $\mathbb C^d$, $d\ge3$, and $Q=I-P$. Conditional on $P$, the ideal references are $M$ independent copies of $P/2$, independent of the single message $S$ and its remote system $R$. A normal message is supported on the range of $P$; the specified anomalous message has reduced state $Q/(d-2)$ on $S$. Quantum receivers (Q) can jointly process the references and message. Measure-and-process receivers (MP) first measure all references with any collective POVM, then apply a quantum instrument to the message, conditioned on the classical record. Only the classical record is retained from the references. Continuous outcomes, adaptive reference-only operations, independent ancillas, and arbitrary classical computation are allowed. The receiver has no description of $P$ or correlated prior calibration, and the message is used once. The budget is the maximum number of references used on any branch of the protocol, including branches that abort or reject.

Write $\alpha_{\rm wc}(P)$ for the largest normal rejection probability and $D(P)$ for the specified anomaly's rejection probability. The accepted map $\mathcal A_P$ is the unnormalized acceptance branch with references and records discarded. Its conditional disturbance is
\begin{equation}
 \delta_g(P)=\sup_{R,\sigma_{SR}}\frac12\left\|
 \frac{(\mathcal A_P\otimes\id_R)(\sigma_{SR})}{p_a(P,\sigma)}-\sigma_{SR}\right\|_1,
 \label{eq:delta}
\end{equation}
where the supremum covers all normal states and remote extensions, with acceptance probability $p_a(P,\sigma)>0$. When the normal acceptance probability is input independent, this distance equals half the diamond-norm distance~\cite{Watrous2018} between the normalized channel and the identity channel on the encoding subspace. The uniform task requires
\begin{equation}
 \alpha_{\rm wc}(P),\delta_g(P)\le\epsilon,\qquad D(P)\ge D_0
 \quad\text{for every }P.\label{eq:task}
\end{equation}
Write $M^*_{\mathcal C}(d,\epsilon,D_0)$ for the minimum such worst-case budget in class $\mathcal C$. Logical errors within $P$ pass this support check.

\emph{A coherent check that preserves normal messages.---}Two reference registers and a normal message share the same two-dimensional support. Their joint state has no fully antisymmetric component: antisymmetrization requires three linearly independent directions in the Hilbert space of each register. A message outside $P$ can supply the third direction. In the resulting binary test, rejection detects this component and acceptance acts as the identity on every normal joint input.

For arbitrary $M$, put $N=M+1$. Let $\Pi_\lambda$ project onto the complete Schur--Weyl sector indexed by a Young diagram $\lambda$, including its permutation multiplicity~\cite{KeylWerner2001}, and let $\ell(\lambda)$ be its number of rows. The binary instrument is
\begin{equation}
 \Pi_{>2}=\sum_{\ell(\lambda)>2}\Pi_\lambda,\qquad
 K_r=\Pi_{>2},\quad K_a=I-\Pi_{>2}.\label{eq:binary}
\end{equation}
The operators are independent of $P$. Every normal joint input lies in the range of $P^{\otimes N}$, which occupies only sectors with at most two rows. Thus $K_a$ preserves every normal joint input and its correlations with untouched $R$, with certain acceptance.

The check reads only an accept/reject flag. Coherently computing the row predicate and uncomputing the label workspace preserves coherence between all accepted sectors. Measuring and forgetting the full Young label of the references and message jointly instead gives a different output channel despite the same acceptance probabilities. For $M=2$, both instruments accept every normal message. For a logical Bell input, the coherent check has fidelity $1$, whereas measuring the full joint Young label gives $2/3$~\cite{Supplemental}.

\emph{Optimal detection from finite references.---}For every integer $M\ge0$, the induced message rejection effect is
\begin{align}
 R_P&=\DQ(M)Q,\nonumber\\
 \DQ(M)&=1-\frac4{M+1}+\frac{M+3}{(M+1)2^M}.\label{eq:exact-law}
\end{align}
False alarm and disturbance vanish. Orthogonal-leakage detection is zero for $M=0,1$, equals $1/12$ at $M=2$, and increases strictly thereafter. The miss probability is $4/M+O(M^{-2})$. A message with leakage weight $\zeta=\Tr(Q\rho_S)$ is rejected with probability $\zeta\DQ(M)$, tending to $\zeta$. For a fixed orthogonal direction, each register is supported on the space spanned by that direction and the range of $P$, of dimension at most three. Covariance makes the response independent of the chosen direction and of the remaining ambient dimensions. Define $M_{\rm Q}^{(0)}(D_0)=\min\{M:\DQ(M)\ge D_0\}$. The minimum budgets at zero false alarm for detection $0.5$, $0.9$, and $0.99$ are $7$, $39$, and $399$, respectively.

This law is optimal among all checks with zero false alarm for every $P$. The span of all normal joint supports comprises the complete sectors with at most two rows. Positivity and zero false alarm for every $P$ force any rejection effect to satisfy $0\le E_r\le\Pi_{>2}$. The instrument in Eq.~\eqref{eq:binary} attains this maximal effect while preserving normal output. The Supplemental Material (SM) gives the sector sum, multiplicities, and padding argument for protocols using at most $M$ references on every branch~\cite{Supplemental}.

\begin{figure*}[t]
\centering\includegraphics[width=\textwidth]{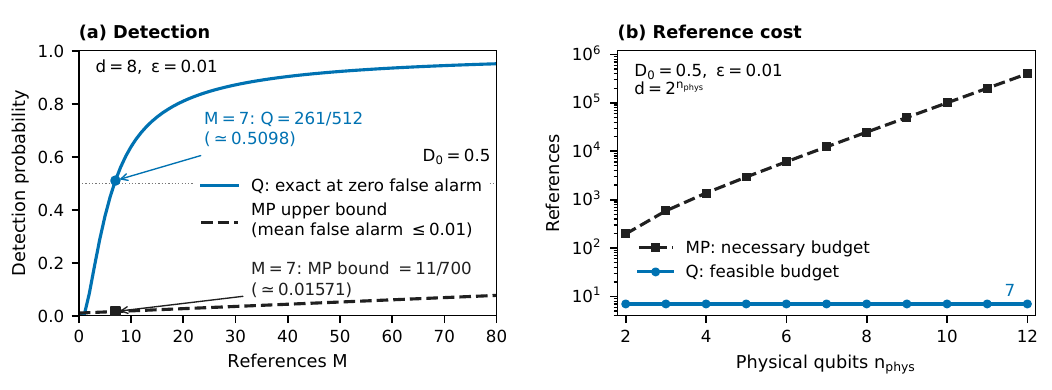}
\caption{Reference costs for an ideal source. (a) Optimal Q detection at zero false alarm from Eq.~\eqref{eq:exact-law}, compared with the Haar-averaged MP upper bound at $d=8$, $\epsilon=0.01$. The latter constrains the guaranteed detection of any MP receiver satisfying the uniform task. At $M=7$, Q detects with probability $261/512$, whereas the MP bound is $11/700$. The Q instrument also satisfies this positive tolerance. (b) At $D_0=0.5$, $\epsilon=0.01$, seven references suffice for Q, while MP requires at least $98(d-2)$ for $d=2^{n_{\rm phys}}$. All curves are evaluated from closed-form expressions.}\label{fig:resources}
\end{figure*}

\emph{The cost of measuring the references first.---}No finite-budget MP receiver can detect the specified orthogonal leakage at zero false alarm for every encoding: $D(P)=0$ for every $P$. The classical record is unrestricted, the references may be measured collectively, and the later message remains quantum. We compare reference ensembles weighted by a message direction's likelihood under the normal and anomalous hypotheses. The resulting posterior operator inequality bounds Haar-averaged detection in terms of Haar-averaged false alarm.

The largest generalized eigenvalue of these posterior operators is
\begin{equation}
 \kappa_M=1+\frac{M}{2(d-2)}
 -\frac{3\,\mathbf1_{M\ {\rm odd}}}{2(d-2)(M+2d-2)}.
 \label{eq:kappa}
\end{equation}
For any MP instrument it gives
\begin{equation}
 \bar D\le\kappa_M\bar\alpha,\qquad
 M^*_{\MP}\ge\left\lceil2(d-2)\left(\frac{D_0}{\epsilon}-1\right)\right\rceil.
 \label{eq:lower}
\end{equation}
Bars denote Haar averages over $P$. Here $\alpha(P)$ is the rejection probability for the normal input $P/2$, so $\alpha(P)\le\alpha_{\rm wc}(P)$. Uniform feasibility requires $\bar\alpha\le\epsilon$, $\bar D\ge D_0$, and hence $\epsilon\kappa_M\ge D_0$. At zero false alarm, averaged detection vanishes and continuity gives pointwise impossibility. The End Matter derives the bound for arbitrary MP receivers without constraining accepted output.

A matching estimate-and-project construction also controls that output. Measure each reference with the covariant rank-one POVM $d|x\rangle\langle x|\,d\nu(x)$, where $\nu$ is Haar probability measure, and let $\widehat P$ project onto the subspace corresponding to the two largest eigenvalues of the empirical covariance $S_M=(d/M)\sum_i|x_i\rangle\langle x_i|$. Its support risk obeys $s=1-\mathbb E_{\mathcal D|P}\Tr(P\widehat P)/2\le C_1d/M$ for $M\ge d$, with finite universal $C_1\ge1$~\cite{GutaEtAl2020,Supplemental}; $\mathbb E_{\mathcal D|P}$ averages over the reference measurement record $\mathcal D$ conditional on $P$. Apply the projective instrument $\{\widehat P,I-\widehat P\}$, accept and return the $\widehat P$ branch, and discard the record. The averaged, unnormalized accepted map is $\mathcal A_P(X)=\mathbb E_{\mathcal D|P}[\widehat P X\widehat P]$. Conditioning on acceptance after this averaging gives
\begin{equation}
 \alpha_{\rm wc}=s,\qquad\delta_g\le s,\qquad
 D=1-\frac{2s}{d-2}.\label{eq:upper}
\end{equation}
These bounds hold for every normal state and its remote extensions. The SM proves the concentration and conditional-channel bounds.

For each fixed $D_0\in(0,1)$ and
\begin{equation}
 0<\epsilon\le\min\{1/32,D_0/2,(1-D_0)/2\},\label{eq:range}
\end{equation}
the uniform task with ideal references obeys
\begin{equation}
 M^*_{\rm Q}\le M_{\rm Q}^{(0)}(D_0),\qquad
 M^*_{\MP}=\Theta_{D_0}(d/\epsilon).\label{eq:resource}
\end{equation}
The constants in $\Theta_{D_0}$ depend only on the fixed target $D_0$. The constructive upper budget is $\lceil C_1d/\epsilon\rceil$. Fig.~\ref{fig:resources} compares the coherent detection law with the MP necessary bound. At $D_0=0.5$, $\epsilon=0.01$, seven coherent references suffice, while MP requires at least $98(d-2)$, or $588$ at $d=8$. With $d=2^{n_{\rm phys}}$, this sample gap is exponential in the number of physical qubits. The matching bounds establish that the linear dependence on $d$ is intrinsic to measurement-first reference access for this detection-and-output task.

\begin{figure*}[t]
\centering\includegraphics[width=\textwidth]{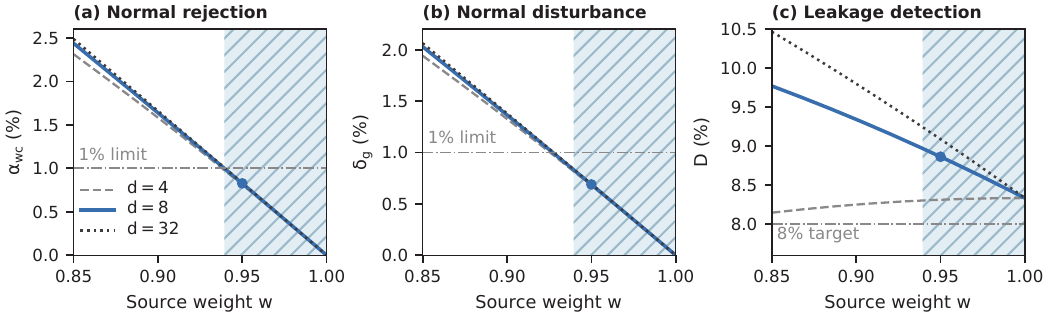}
\caption{Two-reference performance with source contamination. For $d=4,8,32$, normal rejection probability, conditional disturbance of normal messages, and orthogonal-leakage detection are compared with their task thresholds. Hatching marks the feasible interval at $d=8$; dots indicate $w=0.95$. The source weight $w=0.95$ meets $\alpha_{\rm wc},\delta_g\le0.01$ and $D\ge0.08$ for all $d\ge4$; the marked $d=8$ instance has $M^*_{\rm Q}\le2$ and $M^*_{\MP}\ge84$. The curves are evaluated from the channel formulas for contaminated references and ideal receiver operations.}\label{fig:source}
\end{figure*}

\emph{Imperfect references and implementation.---}Consider two independent references in the state $\rho_P(w)=wP/2+(1-w)Q/(d-2)$. For every $d\ge4$, the same binary antisymmetrizer at $w=0.95$ rejects normal messages with probability below $0.84\%$, induces conditional disturbance below $0.70\%$, and detects orthogonal leakage with probability at least $8.31\%$. For this contaminated source, we use the target $D\ge0.08$ and require $\alpha_{\rm wc},\delta_g\le0.01$, obtaining
\begin{equation}
 M^*_{\rm Q}\le2,\qquad M^*_{\MP}\ge14(d-2).\label{eq:noisy}
\end{equation}
The MP lower bound is inherited from the ideal case: because contamination is a $P$-independent channel, any measurement-first protocol on contaminated references can be simulated by applying that channel to ideal references and then running the protocol. At $d=8$, two Q references suffice while MP needs at least 84. Fig.~\ref{fig:source} shows that all three criteria are met simultaneously. Contamination can increase anomalous rejection while also increasing normal rejection and disturbance, so detection alone does not measure the quality of the check.

These source-contamination results assume ideal receiver operations. Coherent Schur-transform circuits~\cite{BurchardtEtAl2025,BaconChuangHarrow2006} approximate the binary instrument with $\operatorname{poly}(M+1,\log d,\log(1/\eta_{\rm S}))$ gates at transform precision $\eta_{\rm S}$. The SM gives the dimension regimes and error bounds for the complete instrument. For power-of-two $d$, the raw references occupy $M\log_2 d$ qubits, excluding the message and workspace. This storage scales as $O(\log d)$ at a fixed ideal-source detection target.

\emph{Discussion.---}Quantum references enable a support check without a classical reconstruction of the encoding. For a single message with no prior calibration, we determine the optimal zero-false-alarm detection law for ideal references and the measurement-first reference cost at a fixed detection target and small normal tolerance. The coherent instrument preserves every normal message and its remote correlations, while any finite-budget measurement-first receiver has zero detection probability at zero false alarm for every encoding.

Exact compression of identically prepared mixed states permits removal of the maximally mixed Schur multiplicity registers~\cite{YangChiribellaEbler2016}. A natural next question is how much quantum memory must be retained after reference preprocessing to support the later check. Task-specific compression may retain less information than faithful reconstruction; the tradeoff involves reference number, retained memory, ambient dimension, and normal-output tolerance. Repeated use raises a complementary question: whether the advantage survives as anomalous messages disturb the references and create correlations, and as classical calibration costs are amortized over multiple messages. 

\begin{acknowledgments}
\FundingStatement
\end{acknowledgments}
\section*{Data Availability}
\DataCodeAvailabilityStatement
\nocite{Krovi2019,NielsenChuang1997,BisioDAranoPerinottiSedlak2011,MoChiribella2019,FanizzaMariGiovannetti2019,Sentis2012,DusekBuzek2002,BergouHillery2005,Sentis2010,SpectralAnomaly2026,Vershynin2018,Schumacher1996}
\bibliography{references}
\clearpage
\section*{End Matter}
\emph{Posterior bound for collective measurements.---}For a unit message direction $v$, define the reference operators weighted by the normal ($g$) and anomalous ($b$) likelihoods:
\begin{align}
 \Omega_g(v)&=\mathbb E_P(P/2)^{\otimes M}\braket{v|P|v}/2,\nonumber\\
 \Omega_b(v)&=\mathbb E_P(P/2)^{\otimes M}\braket{v|Q|v}/(d-2).
 \label{eq:posteriors}
\end{align}
Here $\mathbb E_P$ averages over rank-two subspaces with normalized Haar measure and $Q=I-P$. Both operators have trace $1/d$. With $T_M=\mathbb E_P(P/2)^{\otimes M}$, positivity of both weights Haar-almost everywhere gives the common kernel $\ker T_M$.

The eigenvalue of $T_M$ in an active reference sector $\lambda=(a,b,0,\ldots)$, $a+b=M$, is
\begin{equation}
 t_{a,b}=\frac{a-b+1}{2^M\dim U_{a,b}^{(d)}}.
\end{equation}
The stabilizer of $v$ acts as $U(d-1)$ on $v^\perp$. Its branches $(u,v',0,\ldots)$ obey $b\le u\le a$, $0\le v'\le b$. Each retains Specht multiplicity $f_{a,b}$, the dimension of its irreducible permutation representation. Gelfand--Tsetlin branching~\cite{GrinkoBurchardtOzols2023} gives the normal posterior eigenvalue $B$:
\begin{align}
 x:=\frac{B}{t_{a,b}}
 &=\frac{(a+1-u)(a+2-v')}{2(z+1)(a+d)}\nonumber\\
 &\quad+\frac{(u-b)(b+1-v')}{2(z+1)(b+d-1)},
 \qquad z=a-b.\label{eq:em-branch}
\end{align}
The reference moment vanishes in sectors with a third row; the anomalous eigenvalue is $(t_{a,b}-2B)/(d-2)$, with ratio $(1/x-2)/(d-2)$ to $B$. For fixed $a,b,u$, $x$ decreases with $v'$. At $v'=b$,
\begin{equation}
 x(u+1,b)-x(u,b)=-\frac{b+d-2}{2(a+d)(b+d-1)}<0.
\end{equation}
The ratio is maximal at $(u,v')=(a,b)$. Substituting $a=(M+z)/2$, $b=(M-z)/2$ gives
\begin{equation}
 \kappa_{M,z}=1+\frac{M}{2(d-2)}
 -\frac{z(z+2)}{2(d-2)(M+2d-2)}.
\end{equation}
Its maximum at $z=M\bmod2$ proves Eq.~\eqref{eq:kappa}, including $M=0$, and $\Omega_b(v)\preceq\kappa_M\Omega_g(v)$.

For any collective reference POVM $E(dy)$, the message rejection effect conditioned on $y$ has spectral decomposition $R_y=\sum_j r_{yj}|v_{yj}\rangle\langle v_{yj}|$. Then
\begin{equation}
 \bar\alpha=\int\sum_j r_{yj}\Tr[E(dy)\Omega_g(v_{yj})],
\end{equation}
where $\alpha(P)$ is the rejection probability for $P/2$; replacing $\Omega_g$ by $\Omega_b$ gives $\bar D$. Positivity yields $\bar D\le\kappa_M\bar\alpha$. Uniform feasibility gives $\alpha(P)\le\alpha_{\rm wc}(P)\le\epsilon$ and $\bar D\ge D_0$.

Zero false alarm for every $P$ forces $\bar D=0$. Pad the receiver to $M$ references by supplying and ignoring any unused copies. Its fixed rejection effect $E$ gives
\begin{equation}
 D(P)=\Tr\!\left[E\left((P/2)^{\otimes M}\otimes\frac{I-P}{d-2}\right)\right].
\end{equation}
For a fixed receiver, $D(P)$ is continuous and nonnegative. Full support of Haar measure therefore implies $D(P)=0$ for every $P$.

\onecolumngrid
\clearpage
\setcounter{section}{0}
\setcounter{subsection}{0}
\setcounter{equation}{0}
\setcounter{figure}{0}
\setcounter{table}{0}
\setcounter{secnumdepth}{2}
\renewcommand{\theequation}{S\arabic{equation}}
\renewcommand{\thefigure}{S\arabic{figure}}
\renewcommand{\thetable}{S\arabic{table}}
\renewcommand{\theHequation}{supp.\arabic{equation}}
\renewcommand{\theHfigure}{supp.\arabic{figure}}
\renewcommand{\theHtable}{supp.\arabic{table}}
\renewcommand{\theHsection}{supp.\arabic{section}}
\renewcommand{\theHsubsection}{supp.\arabic{section}.\arabic{subsection}}
\renewcommand{\id}{\operatorname{id}}
\pdfbookmark[0]{Supplemental Material}{supplement-start}
\begin{center}
{\large\bfseries Supplemental Material}\\[0.5em]
{\bfseries Exponential Advantage of Quantum over Classical References in Leakage Detection}
\end{center}
\section{Task, access order, and conditional normalization}\label{supp:sm:task-access}
Let $P$ be an unknown rank-two projector on $\mathbb C^d$, $d\ge3$, and
write $Q=I-P$. The receiver has $M$ independent ideal reference registers in
$(P/2)^{\otimes M}$ and one independent message register $S$. A normal
message--remote-system state $\omega_{SR}$ obeys
$(P\otimes I_R)\omega_{SR}(P\otimes I_R)=\omega_{SR}$; $R$ is untouched.
The specified anomalous message is $Q/(d-2)$. The references are independent
of the message and $R$, conditional on $P$. There is no earlier record or
auxiliary state correlated with $P$.

For a complete receiver instrument, let $\calA_P$ and $\calR_P$ denote the
unnormalized message maps after reference registers and internal records are
discarded. Define
\begin{align}
 \alpha(P)&=\Tr\calR_P(P/2),\qquad
 D(P)=\Tr\calR_P\bigl(Q/(d-2)\bigr),\\
 \alpha_{\rm wc}(P)&=\sup_{\sigma=P\sigma P}\Tr\calR_P(\sigma),\\
 \delta_g(P)&=\sup_{R,\omega_{SR}}\frac12
 \left\|\frac{(\calA_P\otimes\id_R)(\omega_{SR})}
 {\Tr(\calA_P\otimes\id_R)(\omega_{SR})}-\omega_{SR}\right\|_1 .
 \label{supp:eq:conditional-distance}
\end{align}
The supremum is over normal density operators, and the task requires positive
acceptance probability for them. Unless it is independent of the input,
the normalized transformation is not a linear channel; we then use the
statewise definition in Eq.~\eqref{supp:eq:conditional-distance}, not a diamond
norm of that transformation. For a state-independent acceptance probability,
the same supremum is the half-diamond distance of the normalized channel
restricted to the logical input $P$.

The quantum class permits arbitrary joint processing of the references and
message. The measure-and-process (MP) class first measures all references,
using any collective POVM (including continuous records and adaptive
reference-only operations), and then applies an arbitrary record-conditioned
quantum instrument to the message. No $P$-correlated quantum reference
information remains after that measurement, and message-to-reference feedback
is outside this access order. Both classes allow independent ancillas. The
budget is the maximum reference charge on every branch, not its expectation.
Uniform requirements apply to every $P$. Haar integration below is a proof
device; the single-instance task assumes no correlated history or prior
classical calibration record, rather than physical resampling of $P$ between
messages.

\section{Ideal coherent check at arbitrary reference budget}\label{supp:sm:ideal-arbitrary-m}
Put $N=M+1$ for the total number of supplied registers. Under Schur--Weyl duality, write
\begin{equation}
 (\mathbb C^d)^{\otimes N}
 =\bigoplus_{\substack{\lambda\vdash N\\\ell(\lambda)\le d}}
 U_\lambda^{(d)}\otimes[\lambda],\qquad
 \Pi_{>2}=\sum_{\ell(\lambda)>2}\Pi_\lambda .
 \label{supp:eq:schur}
\end{equation}
Here $[\lambda]$ is the entire symmetric-group multiplicity space, of
dimension $f_\lambda$. The row-cutoff predicate is the standard one-sided
rank-testing predicate of O'Donnell and Wright~\cite{ODonnellWright2017} (Sec.~6.1, Proposition~6.1 in the author-posted preprint linked in Sec.~\ref{supp:sm:literature}); the positivity/support method also occurs in universal
unambiguous state comparison~\cite{BarnettCheflesJex2003,JexAnderssonChefles2003}. We apply it to the present
mixed-reference, single-message family and evaluate its exact instrument
performance.

\begin{theorem}[Exact ideal-reference law]\label{supp:thm:ideal}
For every integer $M\ge0$ and every $d\ge3$, the complete binary instrument
\begin{equation}
 K_r=\Pi_{>2},\qquad K_a=I-\Pi_{>2}
 \label{supp:eq:instrument}
\end{equation}
has $\alpha_{\rm wc}(P)=\delta_g(P)=0$ for every $P$. Its induced message
rejection effect is
\begin{equation}
 R_P=\Qzero(M)Q,\qquad
 \boxed{\Qzero(M)=1-\frac4{M+1}+\frac{M+3}{(M+1)2^M}}.
 \label{supp:eq:closed-law}
\end{equation}
Consequently, any message with leakage weight
$\zeta=\Tr(Q\rho_S)$ is rejected with probability $\zeta\Qzero(M)$.
Among all complete instruments with maximum reference budget $M$ and zero
normal false alarm for every $P$, this detection probability is optimal for
each $P$ and each anomalous input, even without imposing a normal-output
constraint on competing instruments. The optimum is attainable while
preserving every normal message and all its remote correlations exactly.
Optimality here is restricted to zero false alarm, not positive tolerance.
\end{theorem}

\begin{lemma}[Complete normal support]\label{supp:lem:support}
The span over rank-two $P$ satisfies
\begin{equation}
 \mathcal K_N:=\operatorname{span}_P\operatorname{ran}P^{\otimes N}
 =\bigoplus_{\ell(\lambda)\le2} U_\lambda^{(d)}\otimes[\lambda].
 \label{supp:eq:support}
\end{equation}
\end{lemma}
\begin{proof}
For a fixed two-plane $P_0$, the restriction to $P_0^{\otimes N}$ is
$\bigoplus_{\ell(\lambda)\le2}U_\lambda^{(2)}\otimes[\lambda]$.
In particular, every such $\lambda$ contains the full multiplicity space
$[\lambda]$. Twirl $(P_0/2)^{\otimes N}$ by $U^{\otimes N}$ with Haar
probability measure on $U(d)$. On the $\lambda$ block the result is
\begin{equation}
 T_N\big|_\lambda=
 \frac{\dim U_\lambda^{(2)}}{2^N\dim U_\lambda^{(d)}}
 I_{U_\lambda^{(d)}}\otimes I_{[\lambda]}
 \quad (\ell(\lambda)\le2),
 \label{supp:eq:twirl}
\end{equation}
and zero otherwise. The coefficient is strictly positive for every
$\ell(\lambda)\le2$. The support of an average of positive operators is
the span of their supports (continuity extends the statement from
almost every rotation to every rotation). Thus $\operatorname{supp}T_N$
is $\mathcal K_N$, proving Eq.~\eqref{supp:eq:support} with all multiplicities.
\end{proof}

\begin{proof}[Proof of Theorem~\ref{supp:thm:ideal}]
\emph{All-instrument bound and normal output.}
Any complete instrument induces a rejection effect $0\le E\le I$ on the
$N$ supplied registers after independent ancillas and internal outcomes have
been eliminated. Zero rejection on every normal message implies, in
particular,
\begin{equation}
 \Tr E(P/2)^{\otimes N}=0\quad\text{for every }P.
\end{equation}
Positivity implies $E^{1/2}P^{\otimes N}=0$. Lemma~\ref{supp:lem:support}
therefore gives
\begin{equation}
 E=\Pi_{>2}E\Pi_{>2},\qquad 0\le E\le\Pi_{>2}.
 \label{supp:eq:effect-bound}
\end{equation}
No covariance or restriction on the competitor's internal measurement is
used. A strategy that adaptively requests at most $M$ independent references
can be supplied with $M$ registers at the outset, using its original causal
order and ignoring unused registers. It is again a complete instrument on
the fixed input. Thus Eq.~\eqref{supp:eq:effect-bound} applies to such strategies
and to independent ancillas of any size.

The instrument in Eq.~\eqref{supp:eq:instrument} is complete because the two
Kraus operators are complementary orthogonal projectors. Its accepted
operator is the identity on $\mathcal K_N$. For every normal $\omega_{SR}$,
\begin{equation}
 (K_a\otimes I_R)
 \bigl[(P/2)^{\otimes M}\otimes\omega_{SR}\bigr]
 (K_a\otimes I_R)
 =(P/2)^{\otimes M}\otimes\omega_{SR}.
 \label{supp:eq:remote-identity}
\end{equation}
This proves exact normal preservation, with acceptance probability one,
before or after discarding the references. Only the binary predicate is
read out. A Schur implementation must uncompute any label workspace;
measuring the full Young label and then forgetting it need not realize
Eq.~\eqref{supp:eq:remote-identity}.

\emph{Reduction to one orthogonal direction.}
The induced positive effect can be written
\begin{equation}
 R_P=\Tr_{\rm ref}\!\left[
 (\tau_P^{1/2}\otimes I)\Pi_{>2}(\tau_P^{1/2}\otimes I)\right],
 \qquad \tau_P=(P/2)^{\otimes M},
 \label{supp:eq:induced-effect}
\end{equation}
where $\tau_P^{1/2}=2^{-M/2}P^{\otimes M}$.
Since $\Pi_{>2}$ commutes with every $U^{\otimes N}$, $R_P$ commutes with
the stabilizer $U(P)\times U(Q)$. Hence $R_P=c_PP+c_QQ$.
Normal rejection vanishes, so $c_P=0$. The coefficient $c_Q$ can be computed
using any unit vector $v\in Q$, rather than averaging over $Q/(d-2)$.
Only $P\oplus\operatorname{span}\{v\}$ is then occupied, reducing the
calculation to three dimensions without changing the Schur projectors on
this support.

\emph{Sector weights, independently derived.}
For $a=M-b$ and $0\le b\le\lfloor M/2\rfloor$, Schur--Weyl duality and
the hook-length formula give
\begin{equation}
 f_{a,b}=\frac{a-b+1}{a+1}\binom Mb,\qquad
 p_{a,b}=\frac{f_{a,b}(a-b+1)}{2^M},\qquad \sum_b p_{a,b}=1.
 \label{supp:eq:weights}
\end{equation}
To obtain the addition probability without assuming a Clebsch--Gordan
coefficient, choose $P=\operatorname{span}(e_1,e_2)$ and $v=e_3$. Define
\[
A(t)=\operatorname{diag}(1/2,1/2,t,0,\ldots,0).
\]
Permutation symmetry
of the central Schur projector implies
\begin{equation}
 [t]\Tr\Pi_\nu A(t)^{\otimes N}
 =N\Tr\Pi_\nu\bigl[(P/2)^{\otimes M}\otimes|v\rangle\langle v|\bigr].
 \label{supp:eq:coefficient}
\end{equation}
Schur--Weyl duality makes the left side $f_\nu[t]s_\nu(1/2,1/2,t)$.
For a shape with more than two rows, a semistandard tableau with exactly
one entry $3$ exists precisely for $\nu=(a,b,1)$, $b\ge1$: that entry
occupies the unique third-row box. The remaining tableau has shape $(a,b)$.
The coefficient is consequently
$s_{a,b}(1/2,1/2)=(a-b+1)/2^M$. The hook-length formula also gives
\begin{align}
 f_{a,b,1}&=\frac{(M+1)!(a-b+1)(a+1)b}{(a+2)!(b+1)!},\nonumber\\
 \frac{f_{a,b,1}}{(M+1)f_{a,b}}&
 =\frac{b(a+1)}{(a+2)(b+1)}=:q_3(a,b).
 \label{supp:eq:branch}
\end{align}
The only at-most-two-row predecessor of $(a,b,1)$ is $(a,b)$, so this
ratio is indeed its conditional addition probability. For $b=0$ there is
no such successor and $q_3=0$. Summing Eq.~\eqref{supp:eq:coefficient}, including
the entire multiplicity $f_\nu$, proves
\begin{equation}
 c_Q=D(M)=\sum_{b=0}^{\lfloor M/2\rfloor}p_{M-b,b}q_3(M-b,b).
 \label{supp:eq:weighted-law}
\end{equation}
In particular, the largest individual branch probability is not the
unconditional detection probability.

\emph{Closed sum and endpoints.}
Since $1-q_3=(M+2)/[(a+2)(b+1)]$, put $L=M+3$, $t=b+1$. Direct factorial
cancellation in Eqs.~\eqref{supp:eq:weights}--\eqref{supp:eq:branch} yields
\begin{align}
 p_{a,b}(1-q_3)
 &=\frac{(L-2t)^2\binom Lt}{L(M+1)2^M},\\
 1-D(M)&=\frac{1}{L(M+1)2^M}
 \sum_{t=1}^{\lfloor(L-1)/2\rfloor}(L-2t)^2\binom Lt.
 \label{supp:eq:miss-sum}
\end{align}
The full binomial second moment is
\begin{equation}
 \sum_{t=0}^L(L-2t)^2\binom Lt=L2^L.
\end{equation}
The summands at $t$ and $L-t$ agree, the two endpoint terms are each $L^2$,
and the middle term is zero when $L$ is even. Thus the sum in
Eq.~\eqref{supp:eq:miss-sum} equals $L2^{L-1}-L^2$. Division gives
\begin{equation}
 1-D(M)=\frac4{M+1}-\frac{M+3}{(M+1)2^M},
\end{equation}
including $M=0,1$. This proves Eq.~\eqref{supp:eq:closed-law}, the induced effect,
and pointwise optimality from Eq.~\eqref{supp:eq:effect-bound}.
\end{proof}

\begin{corollary}[Detection targets and strictly zero-false-alarm budget]
For $0<D_0<1$, define
\[
M_{\rm Q}^{(0)}(D_0)=\min\{M:\Qzero(M)\ge D_0\}.
\]
This is the exact optimum reference budget with zero normal
false alarm and exact normal preservation. It is independent of $d$, and
\begin{equation}
 M_{\rm Q}^{(0)}(1/2)=7,\qquad
 M_{\rm Q}^{(0)}(9/10)=39,\qquad
 M_{\rm Q}^{(0)}(99/100)=399.
 \label{supp:eq:budgets}
\end{equation}
Moreover $M_{\rm Q}^{(0)}(1-\beta)=\Theta(1/\beta)$ and
$\lim_{\beta\downarrow0}\beta M_{\rm Q}^{(0)}(1-\beta)=4$.
\end{corollary}
\begin{proof}
The boundary values are $D(0)=D(1)=0$, $D(2)=1/12$. Subtraction gives
\begin{equation}
 D(M+1)-D(M)=
 \frac{8\,2^M-M^2-5M-8}{2^{M+1}(M+1)(M+2)}.
 \label{supp:eq:monotone}
\end{equation}
For $a_M=M^2+5M+8$, $a_0=8$ and
$2a_M-a_{M+1}=(M+1)(M+2)>0$. Thus $a_M<8\,2^M$ for every $M\ge1$,
proving strict increase from $M=1$ onward. Equation~\eqref{supp:eq:budgets}
follows by checking the displayed budget and its predecessor with rational
arithmetic. For example, $D(6)=201/448<1/2<D(7)=261/512$;
$D(39)>9/10$ whereas $D(38)<9/10$ is equivalent to $410<2^{38}$;
and $D(399)>99/100$ whereas $D(398)<99/100$ is equivalent to
$40100<2^{398}$.

The miss probability is positive for every finite $M$ and obeys
\begin{equation}
 \frac1{M+1}\le1-D(M)\le\frac4{M+1},\qquad
 1-D(M)=\frac4M+O(M^{-2})\quad(M\to\infty).
\end{equation}
The first bounds use $(M+3)/2^M\le3$. Monotonicity, the asymptotic
formula, and the defining inequalities at the minimizing integer and its
predecessor yield the budget limit. In particular, detection exactly one
is impossible at finite budget under the zero-false-alarm requirement.
\end{proof}

\subsection{Finite-precision coherent implementation}\label{supp:sm:schur-implementation}

For $N=M+1$ input registers of local dimension $d>N$, the corrected high-dimensional Schur circuit of Burchardt et al.~\cite[Theorem~1 and Table~1]{BurchardtEtAl2025} has gate complexity $\widetilde O(N^4)$; the suppressed factors are polynomial in $\log N$, $\log d$, and $\log(1/\eta_{\rm S})$. Their construction corrects the multiplicity-space basis change in Krovi's final step~\cite{Krovi2019}, retaining a coherent reversible transform. For $d\le N$, the original BCH complexity $\operatorname{poly}(N,d,\log(1/\eta_{\rm S}))$~\cite{BaconChuangHarrow2006} is already polynomial in $N$ and $\log(1/\eta_{\rm S})$. Combining these regimes gives a coarse $\operatorname{poly}(N,\log d,\log(1/\eta_{\rm S}))$ circuit bound for the present use.

This is a universal qubit-circuit compilation statement, with each qudit represented on $\lceil\log_2 d\rceil$ qubits and polynomial workspace; it imposes no geometric locality constraint or hardware error model. The numerical-gate compilation convention in Sec.~1.3, footnote~1 of Ref.~\cite{BurchardtEtAl2025} uses Solovay--Kitaev synthesis. To state the precision needed here, assign operator-norm error at most $\eta_{\rm S}/L$ to each of $L$ elementary operations in a transform circuit. Telescoping yields transform error at most $\eta_{\rm S}$ on the valid working space, with polynomial logarithmic synthesis overhead. This worst-case error allocation, rather than an average-input fidelity, is what enters the instrument bound below.

Let $U_{\rm S}$ denote a coherent Schur transform on $N=M+1$ registers,
including its workspace. The label operation $C_f$ copies only
$f(\lambda)=\mathbf1_{\ell(\lambda)>2}$ to a binary flag, using reversible
classical computation. With clean workspace and flag initialized by the
isometry $J$, the ideal dilation is
\begin{equation}
 V=U_{\rm S}^{\dagger}C_fU_{\rm S}J.
\end{equation}
The complete Young label is never measured. After the inverse transform,
only the flag is measured, and the references and workspace are traced out.
This produces the two Kraus branches in Eq.~\eqref{supp:eq:instrument}.

Write $\eta_{\rm S}$ for transform implementation accuracy, separately from
the normal tolerance $\epsilon$ and the source impurity $1-w$.
If the coherent transform and its inverse each have operator-norm error at
most $\eta_{\rm S}$ on their valid working spaces, telescoping the two
passes gives
\begin{equation}
 \|\widetilde V-V\|_\infty\le2\eta_{\rm S}.
 \label{supp:eq:schur-dilation-error}
\end{equation}
The flag computation preserves the valid Schur space. For a unitary circuit
approximating the forward isometry on clean inputs, its reversed circuit
has the same inverse error on the exact image: this follows by multiplying
$\widetilde U_{\rm S}-U_{\rm S}$ by $\widetilde U_{\rm S}^{\dagger}$.
Thus the inverse pass in Eq.~\eqref{supp:eq:schur-dilation-error} acts on the
required domain. The binary predicate can be computed exactly with
reversible gates; errors in any additional compiled operations must instead
be included in the total instrument error.

For isometries, the half-diamond distance of their channels is bounded by
their operator-norm distance, also on inputs entangled with an arbitrary
remote system. Measurement of the flag and partial traces are channels, so
contractivity bounds the complete flag--message instruments by
\begin{equation}
 \mu:=\tfrac12\|\widetilde{\mathcal I}-\mathcal I\|_\diamond
 \le 2\eta_{\rm S}.
 \label{supp:eq:schur-instrument-error}
\end{equation}
This transfers a coherent transform error to the full returned instrument,
rather than only to its decision probabilities.

For an ideal reference source and any complete-instrument error $\mu<1$,
the conditional-normalization bound in Eq.~\eqref{supp:eq:normalization-transfer} gives
\begin{equation}
 \alpha_{\rm wc}\le\mu,\qquad
 D\ge\Qzero(M)-\mu,\qquad
 \delta_g\le\frac{2\mu}{1-\mu}.
 \label{supp:eq:ideal-implementation-tolerance}
\end{equation}
Indeed, every normal input has ideal accepted probability one and exact
return. The noisy acceptance is at least $1-\mu$, while the accepted
subnormalized state differs in trace norm by at most $2\mu$. Normalization
then gives the displayed sufficient bound, uniformly over normal remote
extensions. For $M=7$ and $\mu=1/250$,
\begin{equation}
 \alpha_{\rm wc}\le0.004,\qquad
 \delta_g\le\frac2{249},\qquad
 D\ge\frac{32369}{64000}=0.505765625.
 \label{supp:eq:seven-copy-implementation}
\end{equation}
Consequently the seven-reference, $D_0=0.5$, $\epsilon=0.01$ example has
a nonzero complete-instrument error margin. These are conditional
implementation guarantees; no hardware realization at this $\mu$ is claimed.

\section{Exact collective MP posterior and decision bound}\label{supp:sm:mp-posterior}
\paragraph{General-dimensional ideal posterior law.}
Let $P$ be a Haar rank-two projector in $\mathbb C^d$, $d\ge3$, and $Q=I-P$.
For $M$ references $(P/2)^{\otimes M}$, compare the normal query $P/2$ with
the anomalous query $Q/(d-2)$. The exact collective posterior slope is
\begin{equation}
 \kappa_M^{(d,2)}(1)=1+\frac{M}{2(d-2)}
 -\frac{z(z+2)}{2(d-2)(M+2d-2)}.
\end{equation}
Here $z=M\bmod2$: the final term is zero for even $M$ and has numerator three for odd $M$.
The large-$M$ slope is $1/[2(d-2)]$.

To prove the formula, write an active reference Young label as
$\lambda=(a,b,0,\ldots,0)$, $a\ge b\ge0$, $a+b=M$, and its stabilizer branch as
$\mu=(u,v',0,\ldots,0)$ with $b\le u\le a$, $0\le v'\le b$.
The moment eigenvalue is $t_{a,b}=(a-b+1)/(2^M\dim V_{a,b}^{U(d)})$.
With imbalance $z=a-b$, the three Gelfand--Tsetlin weights are
\begin{align}
 q_1&=\frac{(a+1-u)(a+2-v')}{(z+1)(a+2)},\\
 q_2&=\frac{(u-b)(b+1-v')}{(z+1)(b+1)},\\
 q_3&=\frac{v'(u+1)}{(a+2)(b+1)}.
\end{align}
They sum to one and do not depend on $d$. The third-row moment vanishes, while
$t_{a+1,b}/t_{a,b}=(a+2)/[2(a+d)]$ and
$t_{a,b+1}/t_{a,b}=(b+1)/[2(b+d-1)]$ for valid additions.
Writing $g=B$ and $x=B/t_{a,b}$, the normal posterior eigenvalue has
\begin{equation}
 x=\frac{(a+1-u)(a+2-v')}{2(z+1)(a+d)}
 +\frac{(u-b)(b+1-v')}{2(z+1)(b+d-1)}>0.
\end{equation}
The anomalous eigenvalue is $h=(t_{a,b}-2B)/(d-2)$, hence
$h/g=(1/x-2)/(d-2)$. For fixed $a,b,u$, $x$ strictly decreases with $v'$.
At $v'=b$, its forward difference in $u$ is
$-(b+d-2)/[2(a+d)(b+d-1)]<0$. Thus its unique minimum for fixed $a,b$
occurs at $(u,v')=(a,b)$, with
\begin{equation}
 x_{a,b}=\frac{M+2d-2}{2(a+d)(b+d-1)}.
\end{equation}
Substitute $a=(M+z)/2$, $b=(M-z)/2$. The resulting ratio decreases
strictly with $z\ge0$, leaving $z=M\bmod2$ as the unique maximizer.
Reference sectors with more than two rows have both posterior eigenvalues
zero. The weighted posterior traces are $1/d$ each. Independent exact block
checks verify the branch coefficients and weighted traces for representative
dimensions and reference counts.

For a uniform task with $\alpha_{\rm wc}\le\epsilon$ and $D\ge D_0$, the posterior order requires $\kappa_M^{(d,2)}(1)\epsilon\ge D_0$. The order bounds the Haar-averaged detection of any MP receiver; uniform feasibility makes it a necessary budget condition without any covariance assumption on that receiver.

\paragraph{Zero normal false alarm implies pointwise zero detection.}
If $\alpha_{\rm wc}(P)=0$ for every $P$, then $\bar\alpha=0$ and the finite-budget posterior bound gives $0\le\bar D\le\kappa_M\bar\alpha=0$. Pad any adaptive reference schedule to its finite maximum budget $M$ and absorb its independent ancillas and internal records into its total rejection effect $E$, which is independent of $P$. Then
\begin{equation}
 D(P)=\Tr\left[E\left((P/2)^{\otimes M}\otimes\frac{I-P}{d-2}\right)\right].
 \label{supp:eq:mp-zero-pointwise}
\end{equation}
This is a nonnegative continuous function of $P$. Since Haar measure has full support on the Grassmannian, a positive value anywhere would give a positive integral on an open neighborhood. Thus $D(P)=0$ for every $P$. No covariance or accepted-output condition is imposed on the competing MP receiver.

\subsection{Independent tensor-space verification}
The exact scan uses rational arithmetic and never forms a $d^M$ matrix. To verify the contraction and normalization by a different method, a small calculation expands each Haar moment as
\begin{equation}
 T_k=\sum_{\pi\in S_k}c_\pi V_\pi,
 \qquad \sum_{\pi}d^{c(\sigma^{-1}\pi)}c_\pi=2^{c(\sigma)-k},
\end{equation}
where $c(\pi)$ is the number of cycles. For $k\le d$ the permutation Gram matrix is invertible. Contracting the last register of $T_{M+1}$ with $|0\rangle$ gives $\Omega_g$; the other posterior is $(T_M-2\Omega_g)/(d-2)$. Direct diagonalization on the common support agrees with the branch ratios in small dimensions and also verifies traces, positivity, and kernel agreement. This independent tensor-space calculation supports the contraction formulas; the branch optimization proves the result at arbitrary $M$.

The antisymmetry check uses signed permutation traces for $r=2,3,4$ and direct action on a normal logical Bell input.

\subsection{Contamination as a channel on the references}
For the general source $\rho_P(w)=wP/2+(1-w)Q/(d-2)$ and $2/d\le w\le1$, the depolarizing channel
\begin{equation}
 \mathcal L_w(X)=c_wX+(1-c_w)\Tr(X)I_d/d,
 \quad c_w=\frac{dw-2}{d-2},
\end{equation}
sends $P/2$ to $\rho_P(w)$. As the independent data likelihood is unchanged, $\Omega_{g,b}(w)=\mathcal L_w^{\otimes M}[\Omega_{g,b}(1)]$. Positivity therefore gives $\kappa_M^{(d,2)}(w)\le\kappa_M^{(d,2)}(1)$. This establishes the same necessary MP bound for these contaminated sources. Below we derive the coherent noisy channel at general $d$, followed by detailed dimension-four benchmarks.

\section{Measurement-first estimator and fixed-detection resource law}\label{supp:sm:estimator-resource-law}
Let $M_{\rm Q}^*(d,\epsilon,D_0)$ and $M_{\rm MP}^*(d,\epsilon,D_0)$
be minimum worst-case reference budgets for the uniform ideal-source task
$\alpha_{\rm wc}(P),\delta_g(P)\le\epsilon$, $D(P)\ge D_0$ for all $P$.
The exact MP posterior law is derived in Sec.~\ref{supp:sm:mp-posterior}; the
following estimator argument specifies the upper-bound constants and combines
them with Theorem~\ref{supp:thm:ideal}.

\begin{corollary}[Fixed $D_0$ separation]\label{supp:cor:resource}
For every fixed $D_0\in(0,1)$, every $d\ge3$, and
\begin{equation}
 0<\epsilon\le\min\{1/32,D_0/2,(1-D_0)/2\},
 \label{supp:eq:epsilon-range}
\end{equation}
there is a finite universal $C_1\ge1$ such that
\begin{align}
 M_{\rm Q}^*(d,\epsilon,D_0)&\le M_{\rm Q}^{(0)}(D_0),\\
 \left\lceil2(d-2)\left(\frac{D_0}{\epsilon}-1\right)\right\rceil
 &\le M_{\rm MP}^*(d,\epsilon,D_0)
 \le\left\lceil C_1d/\epsilon\right\rceil .
 \label{supp:eq:two-bounds}
\end{align}
In particular, $M_{\rm MP}^*=\Theta_{D_0}(d/\epsilon)$ uniformly on this
interval. The constants do not depend on $d$ or $\epsilon$.
\end{corollary}
\begin{proof}
The quantum inequality follows by using the zero-false-alarm instrument. The
exact MP posterior slope is, with $z=M\bmod2$,
\begin{equation}
 \kappa_M^{(d,2)}(1)=1+\frac{M}{2(d-2)}
 -\frac{z(z+2)}{2(d-2)(M+2d-2)}.
 \label{supp:eq:kappa}
\end{equation}
For every MP receiver, $\overline D\le\kappa_M\overline\alpha$, where
the bars denote Haar averages of $D(P)$ and
$\alpha(P)=\Tr\calR_P(P/2)$, not of an unspecified query. Uniform
feasibility gives $D_0\le\overline D$ and
$\overline\alpha\le\sup_P\alpha_{\rm wc}(P)\le\epsilon$. Since
$\kappa_M\le1+M/[2(d-2)]$, the necessary budget in
Eq.~\eqref{supp:eq:two-bounds} follows. This decision bound imposes no assumption
on the quality of the accepted output. The parity-refined necessary integer
budget is $\min\{M\ge0:\epsilon\kappa_M^{(d,2)}(1)\ge D_0\}$;
it is not an achievable-budget claim.

For the upper bound, we use the following uniform-POVM estimator argument.
Measure each reference with the POVM
$d|x\rangle\langle x|d\nu(x)$, where $\nu$ is Haar probability measure
on the unit sphere. Under $P/2$ the outcome density is
$w_P(x)=d\langle x|P|x\rangle/2$. Set $y=\sqrt d\,x$ and
$S_M=M^{-1}\sum_i y_i y_i^\dagger$. The two-design identity gives
\begin{equation}
 T_P=\mathbb E S_M=\frac d{d+1}(I+P/2),\qquad
 g=\frac d{2(d+1)}\ge\frac13.
\end{equation}
For every unit vector $u$, Cauchy--Schwarz and the exact Haar tail give
\begin{equation}
 \Pr\{|\langle u,y\rangle|>t\}
 \le\sqrt{\frac32}\exp(-t^2/3),
\end{equation}
using $\mathbb E_\nu w_P^2=3d/[2(d+1)]$. Centered scalar quadratic forms
therefore have a universal subexponential norm. Scalar Bernstein
\cite[Theorem~2.8.1]{Vershynin2018} and a $1/4$ net of the complex unit sphere
with at most $9^{2d}$ points give, for universal positive $c$,
\begin{equation}
 \Pr\{\|S_M-T_P\|_\infty>2t\}
 \le2\exp\{2d\log9-cM\min(t^2,t)\}.
\end{equation}
Integrating after $t$ exceeds a universal multiple of
$\max(\sqrt{d/M},d/M)$ yields a finite universal $C$ with
\begin{equation}
 \mathbb E\|S_M-T_P\|_\infty^2
 \le C(d/M+d^2/M^2).
\end{equation}
Thus this is a dimension-uniform estimate, without a $\log d$ penalty.
The uniform-POVM covariance/net method also appears in
Gu\c{t}\u{a} et al.~\cite[Theorem~5]{GutaEtAl2020}; the calculation above uses the
present biased outcome distribution explicitly.

Let $\widehat P$ be the top-two sample projector and
$\ell=1-\Tr(P\widehat P)/2$. For $M\ge d$ it is equivariant almost
surely: the continuous sample law has a simple covariance spectrum outside
a polynomial null set. With $e=\|S_M-T_P\|_\infty$, the variational
principle gives
\begin{equation}
 2g\ell\le e\|P-\widehat P\|_1\le4e\sqrt\ell,
 \qquad s:=\mathbb E\ell\le C_1d/M \quad(M\ge d)
 \label{supp:eq:support-risk}
\end{equation}
for a finite universal $C_1$, enlarged to be at least one.

The actual accepted map, averaged over the discarded reference record, is
\begin{equation}
 \calA_P(X)=\mathbb E[\widehat P X\widehat P].
\end{equation}
Stabilizer covariance implies input-independent normal acceptance $1-s$.
Its normalized channel has entanglement fidelity~\cite{Schumacher1996}
\begin{equation}
 F=\frac{\mathbb E(1-\ell)^2}{1-s}\ge1-s.
\end{equation}
In the input-conjugate/output representation
$(\overline U_P\otimes U_P)\oplus(\overline U_P\otimes U_Q)$, the only
invariant vector is the logical Bell vector. Hence the normalized Choi
state has this vector as an eigenvector with weight $F$. Subtracting that
component leaves a positive Choi operator with input marginal
$(1-F)I_P/2$. The normalized channel is therefore
$F\id_P+(1-F)\mathcal N$ for a channel $\mathcal N$ with Choi support
orthogonal to the Bell vector. Convexity and the Bell input respectively
upper bound and attain the error, giving
\begin{equation}
 \alpha_{\rm wc}=s,\qquad
 \delta_g=\tfrac12\|\calA_P/(1-s)-\id_P\|_\diamond=1-F\le s,
 \qquad D=1-\frac{2s}{d-2}.
 \label{supp:eq:mp-output}
\end{equation}
These are record-averaged guarantees, not guarantees for every realized
reference record. Choosing $M=\lceil C_1d/\epsilon\rceil$ ensures
$M\ge d$, $s\le\epsilon$, and
$D\ge1-2\epsilon/(d-2)\ge D_0$ on Eq.~\eqref{supp:eq:epsilon-range}.
Finally $\epsilon\le D_0/2$ makes the lower bound at least
$(d-2)D_0/\epsilon\ge(D_0/3)d/\epsilon$, proving the asserted order.
\end{proof}

At $D_0=1/2$ and $\epsilon=1/100$, this gives the concrete uniform ideal
task
\begin{equation}
 M_{\rm Q}^*\le7,\qquad M_{\rm MP}^*\ge98(d-2).
 \label{supp:eq:ideal-example}
\end{equation}
For $d=8$, the MP lower bound in Eq.~\eqref{supp:eq:ideal-example} is $588$;
it is a necessary budget and does not assert an MP construction at $588$.
Seven is a feasible quantum budget at this positive tolerance, not a proved
positive-tolerance optimum. For $\epsilon\ge D_0$, the independent random
rejector $K_r=\sqrt{D_0}I$, $K_a=\sqrt{1-D_0}I$ uses zero references and
has $\alpha_{\rm wc}=D_0$, $\delta_g=0$, $D=D_0$. Both optimum budgets
are then zero, so the small-$\epsilon$ qualification is substantive.
The joint optimal law as $\beta=1-D_0$ and $\epsilon$ both vanish is not
established here. For one logical qubit encoded in $n_{\rm phys}$ physical
qubits, $d=2^{n_{\rm phys}}$; the dimension separation is a reference-sample
separation, not a time-complexity result.

\paragraph{Positive tolerance can reduce a feasible coherent budget.}
For the ideal source, take $M=2$, $p=1/50$, and the three-register antisymmetrizer $\Pi$. The complete Kraus representation
\begin{equation}
 K_a=\sqrt{1-p}(I-\Pi),\qquad
 K_{r,1}=\sqrt p\,I,\qquad K_{r,2}=\sqrt{1-p}\,\Pi
 \label{supp:eq:random-reject-instrument}
\end{equation}
has one acceptance flag and one rejection flag; the two rejection Kraus operators share the latter. Completeness follows from $\Pi^2=\Pi$. On every normal input, including any remote extension, the accepted map is $(1-p)\id$, so the normalized output is unchanged. Its scores are
\begin{equation}
 \alpha_{\rm wc}=\frac1{50},\qquad\delta_g=0,\qquad
 D=p+\frac{1-p}{12}=\frac{61}{600}>\frac1{10}.
\end{equation}
Thus $M_{\rm Q}^*(d,0.02,0.10)\le2$ for every $d\ge3$, whereas the exact zero-false-alarm budget for this ideal-source detection target is three. This construction leaves the positive-tolerance optimum unspecified. The exact three-reference certificate in Sec.~\ref{supp:sm:finite-certificates-hardware} concerns the distinct $d=4,w=0.90$ contaminated-source task.

\section{Five-percent source contamination in all dimensions}\label{supp:sm:source-noise}
Let $d\ge3$, $k=d-2$, and let the independent references be
\(\rho_P=wP/2+(1-w)Q/k\), with $2/d\le w\le1$.  Apply the two-reference
binary antisymmetrizer instrument, $K_a=I-\Pi_A$, $K_r=\Pi_A$, where
\(\Pi_A\) is the full three-register antisymmetrizer.  This general source
model is used for the analytical source-noise curves; the uniform corollary then
specializes to $d\ge4$ and $w=19/20$.

\paragraph{General source contraction.}
For $X=PXP$, the complete accepted map is
\begin{equation}
 \calA_P(X)=c_{\rm id}X+c_P\Tr(X)P/2+c_Q\Tr(X)Q/k,
 \label{supp:eq:general-noisy-map}
\end{equation}
where, with $\eta=1-w$,
\begin{align}
 \alpha&=\eta/6-\eta^2/(6k),&
 g_A&=\eta^2(k-1)/(18k),\nonumber\\
 c_{\rm id}&=1-2\alpha+g_A,&
 c_P&=w\eta/9,\qquad
 c_Q=\eta/18+\eta^2(k-2)/(18k).
 \label{supp:eq:general-noisy-coefficients}
\end{align}
The normal rejection is input independent and equals $\alpha_{\rm wc}=\alpha$;
the accepted probability is $A=1-\alpha$.  The designated anomaly has
\begin{equation}
 D=\bigl(1-w^2/2-\eta^2/k-2\eta/k+2\eta^2/k^2\bigr)/6.
 \label{supp:eq:general-noisy-D}
\end{equation}
To derive the full map, diagonalize $\rho_P$, with eigenvalues $w/2,w/2$ on
$P$ and $\eta/k$ on $Q$.  The diagonal antisymmetrizer contraction on a
message eigenvector of weight $t_i$ is
$(1-\Tr\rho_P^2-2t_i+2t_i^2)/6$, giving $\alpha$ or $D$ above.  Direct
signed-permutation contraction gives
\(\Tr_{12}[\Pi_A(\rho_P^{\otimes2}\otimes X)]=\alpha X\) on normal $X$.
For the remaining sandwich $G(X)=\Tr_{12}[\Pi_A(\rho_P^{\otimes2}\otimes
X)\Pi_A]\), matrix units suffice.  If $i\in\{0,1\}$ and the eigenvalues are
denoted $t_u$, then
\begin{equation}
 G(|i\rangle\langle i|)=\frac19
 \sum_{\substack{u<v\\u,v\ne i}}t_ut_v
 (|u\rangle\langle u|+|v\rangle\langle v|+|i\rangle\langle i|),
 \label{supp:eq:general-noisy-wedge}
\end{equation}
while $G(|0\rangle\langle1|)=g_A|0\rangle\langle1|$.  Each unordered
triple has two ordered reference strings of antisymmetric weight $1/6$ and a
one-register reduction of $1/3$.  Reading the logical and orthogonal
coefficients in Eq.~\eqref{supp:eq:general-noisy-wedge} gives
Eq.~\eqref{supp:eq:general-noisy-map}; the same matrix-unit calculation applies
blockwise to a normal input entangled with an untouched remote system.

The coefficients are nonnegative on the stated source interval.  Decomposing
the logical replacement into four Pauli conjugations, and using a logical Bell
input to attain the resulting convexity bound, gives
\begin{equation}
 \delta_g=\frac{3c_P/4+c_Q}{1-\alpha}.
 \label{supp:eq:general-noisy-delta}
\end{equation}

\begin{corollary}[Uniform five-percent source-contamination instance]
For all $d\ge4$, two independent references from the stated source yield
\begin{align}
 \alpha_{\rm wc}&=\frac{20k-1}{2400k}\le\frac1{120},\\
 \delta_g&=\frac{99k-4}{6(2380k+1)}\le\frac{33}{4760},\\
 D&=\frac{439k^2-82k+4}{4800k^2}\ge\frac{133}{1600}.
 \label{supp:eq:noisy-laws}
\end{align}
For the uniform task $\alpha_{\rm wc},\delta_g\le1/100$, $D\ge2/25$,
\begin{equation}
 M_{\rm Q}^*\le2,\qquad M_{\rm MP}^*\ge14(d-2).
 \label{supp:eq:noisy-budgets}
\end{equation}
Here both budgets refer to this contaminated-source task. The MP inequality
is necessary, not a construction at the indicated budget.
\end{corollary}
\begin{proof}
Substitution of $w=19/20$ into
Eqs.~\eqref{supp:eq:general-noisy-coefficients}--\eqref{supp:eq:general-noisy-delta}
gives the three displayed scores in the corollary.
The uniform bounds follow, for example, from the exact positive differences
\begin{align}
 \frac1{120}-\alpha&=\frac1{2400k},\\
 \frac{33}{4760}-\delta_g&=
 \frac{9619}{14280(2380k+1)},\\
 D-\frac{133}{1600}&=\frac{(20k-1)(k-2)}{2400k^2}.
\end{align}
The first two worst limits occur as $k\to\infty$; the detection minimum
occurs at $k=2$. The thresholds have strictly positive uniform margins,
so the two-reference protocol is feasible for the displayed task.

For the MP bound, the contaminated source is the image of $P/2$ under the
$P$-independent depolarizing channel
\begin{equation}
 \mathcal L_w(X)=cX+(1-c)\Tr(X)I_d/d,
 \qquad c=\frac{dw-2}{d-2}=\frac{19k-2}{20k}\in[0,1].
\end{equation}
Applying $\mathcal L_w^{\otimes M}$ to both ideal reference posterior
operators preserves their positive order. The same ideal slope
$\kappa_M\le1+M/(2k)$ thus remains a valid upper bound for the contaminated
source. Uniform feasibility requires $\kappa_M\ge(2/25)/(1/100)=8$;
therefore $M\ge14k$. In particular, at $d=8$ this gives
$M_{\rm Q}^*\le2$ and $M_{\rm MP}^*\ge84$.
\end{proof}

The contaminated-source corollary has a fixed modest detection target. It
does not extend the ideal arbitrary-$M$ zero-false-alarm result to high detection
under five-percent contamination. Source contamination and gate noise are
different models. An accepted message is consistent with the subspace
check; its logical action inside $P$ is not certified.

\section{Finite-sample estimator diagnostics}\label{supp:sm:estimator-diagnostics}
\paragraph{Finite-sample measurement-first estimator.}
For each dimension we fixed the evaluator frame so that
$P=\operatorname{diag}(1,1,0,\ldots,0)$ and drew independent outcomes of the
rank-one covariant POVM on the reference state $P/2$.  This use of $P$ is
strictly confined to the source/evaluator simulator: the learner receives only
the outcome vectors $x_i$.  Writing $t=\lVert Px\rVert^2$, the exact outcome
law is sampled as $t\sim\operatorname{Beta}(3,d-2)$, followed by independent
complex-Haar directions in $\mathbb C^2$ and $\mathbb C^{d-2}$.  The learner
forms $S_M=(d/M)\sum_i|x_i\rangle\langle x_i|$ and uses its top-two projector
$\widehat P_M$ as the accept effect, with $I-\widehat P_M$ as reject.

For each $d\in\{4,8,16,32,64\}$, we generated 64 independently seeded reference datasets and used the nested prefixes $M/d=2,4,\ldots,256$.  The predeclared
extension $M/d=512,1024$ was run because the initial grid did not cover the
$0.01$ target in all dimensions.  Prefix points within one seed are therefore
correlated; uncertainty bars at each fixed $(d,M)$ instead use the 64
independent datasets (nominal two-sided pointwise Student-$t$ 95\% intervals).

For a fixed learned projector let
$h=\operatorname{Tr}(P\widehat P_M)/2$ and
$e=1-h$.  The population support risk is $s=\mathbb E_{\mathcal D|P}e$, estimated by the dataset mean.  The accepted
logical channel is averaged over all reference records before conditioning.
Under exact population covariance, its normalized Choi fidelity with the ideal
logical identity is evaluated from the sufficient moments as
\begin{equation}
 F_{\rm acc}=\frac{\mathbb E_{\mathcal D|P} h^2}{\mathbb E_{\mathcal D|P} h},\qquad
 \delta_{\rm acc}=1-F_{\rm acc}.
\end{equation}
For the finite empirical ensemble, $1-F_{\rm acc}$ is consequently reported
only as a covariance-based estimate of this population $\delta$, rather than
as an empirical channel distance.  We separately save the full averaged
physical-output Bell Choi matrix
$J=\mathbb E_{\mathcal D|P}\,[\operatorname{vec}(\widehat P_M T)
\operatorname{vec}(\widehat P_M T)^\dagger/2]$, normalize it by its trace,
and directly report $\tfrac12\|J/\operatorname{Tr}J-\Phi_T\|_1$.  Thus the
calculation does not average conditional disturbances record by record.  We
also save $B=P\widehat P_M P$ in logical coordinates for every record and
report $\|\mathbb E_{\mathcal D|P} B-(\mathbb E_{\mathcal D|P} h)I_2\|_\infty/\mathbb E_{\mathcal D|P} h$ as a finite-
Monte-Carlo covariance diagnostic.  No evaluator-dependent twirl is applied
to the receiver.  We report normal acceptance $\mathbb E_{\mathcal D|P} h$ and the derived
orthogonal-anomaly rejection $D=1-2s/(d-2)$, each with pointwise
95\% intervals transformed from the 64 dataset-level support risks.

The threshold summary uses only the measured grid. At each $(d,M)$, a
``failing'' point has the lower pointwise interval endpoint above $\epsilon$;
a ``passing'' point has the upper endpoint at or below $\epsilon$. Other
points are inconclusive. The last failing and first passing grid points are
reported where present; a missing passing point is NOT\_REACHED and a gap
without either endpoint is INCONCLUSIVE. These pointwise classifications
are not a confidence interval for the optimal or continuous reference budget.

\begin{figure}[t]
\centering
\includegraphics[width=0.72\linewidth]{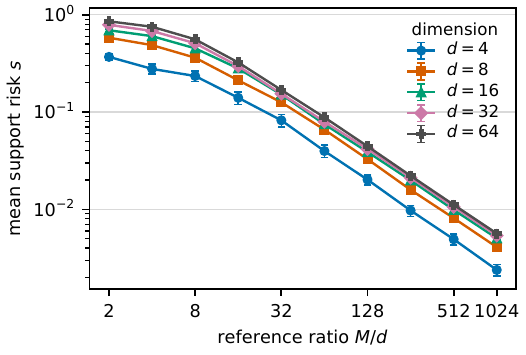}
\caption{Saved estimator support-risk diagnostics. Points show the mean
$s=\mathbb E_{\mathcal D|P}[1-\Tr(P\widehat P_M)/2]$ and bars are nominal
two-sided pointwise Student-$t$ 95\% intervals over 64 independent reference
datasets at fixed $(d,M)$. Curves only connect nested-prefix display points;
they are not fitted scaling laws and do not provide a joint confidence band.}
\label{supp:fig:sm-estimator-risk}
\end{figure}

\section{Storage-interface benchmark and crosstalk diagnostics}\label{supp:sm:application-diagnostics}
\paragraph{Leakage-interface benchmark.}
For every instance the source samples an independent Haar frame $U=(u_1,\ldots,u_d)\in U(d)$, defines $T=(u_1,u_2)$, $P=TT^\dagger$, and reserves $V=(u_3,u_4)$ for the source-defined fault. The coherent receiver receives two message-independent mixed references $\rho_{P,w}=wP/2+(1-w)(I-P)/(d-2)$ and one physical query, but never $U,T,P$, or $V$. We use $w\in\{1,0.95\}$ and the coherent cross-talk $\exp[-i\theta(VT^\dagger+TV^\dagger)]$, with $\sin^2\theta=\zeta\in\{0,0.1,0.25,0.5,1\}$.

The coherent receiver implements the complete binary instrument $K_{\rm a}=I-\Pi_A$ and $K_{\rm r}=\Pi_A$ on both references and the query. Exact low-rank Kraus contractions propagate both branches and the remote qubit $R$; accepted and rejected $R$--output Choi states are saved. The normal $\alpha_{\rm wc}$ is reconstructed from the full logical accept effect. Bell, nonmaximally entangled, and six-pure-state distances are finite-state, reference-extended diagnostics, not diamond norms.

For MP, covariant rank-one measurements yield the top-two eigenspace $\widehat P$ of $S_M=(d/M)\sum_i|x_i\rangle\langle x_i|$, with actual $M=2,8,32,128,512$. The exact Born sampler draws $t\sim\mathrm{Beta}(3,d-2)$ with probability $w$ and $t\sim\mathrm{Beta}(2,d-1)$ otherwise. After references are consumed, the physical query instrument is $\widehat P(\cdot)\widehat P$ and $(I-\widehat P)(\cdot)(I-\widehat P)$. We average unnormalised Choi channels over eight independent frames/datasets in an evaluator-only canonical frame before conditional distances are evaluated; recordwise distances are diagnostics only.

With fault fraction 0.1, $A=0.9A_0+0.1A_1$ and $G=0.9G_0+0.1G_1$, where $G$ is unnormalised overlap with the intended clean source-reference target; $G/A$ is reported separately. Known-$P$ is an explicitly privileged oracle. Logical $X,Z$ controls remain in $P$ and delimit the task as leakage detection.

\paragraph{Matched-ensemble converse.}
The coherent fault defined above is defined by a full Haar frame, including the partner modes $V$. Averaging their phases eliminates the $P$--$Q$ coherences and averaging their directions gives the anomalous marginal
$(1-\zeta)P/2+\zeta Q/(d-2)$. Therefore the decision converse becomes
\begin{equation}
 \bar D_\zeta\le[1+\zeta(\kappa_M-1)]\bar\alpha.\label{supp:eq:haar-crosstalk}
\end{equation}
Equation~\eqref{supp:eq:haar-crosstalk} concerns the matched full-Haar ensemble and any task uniform over its frames. For a fixed frame or a structured encoder distribution, its likelihood operators must be evaluated separately. The posterior inequality supplies the full-class bound; the simulated MP points characterize the specified estimator.

\paragraph{Estimator diagnostics and fault traffic.}
The saved $d=8,w=1$ MP estimator with $M=2$ has normal worst-case
rejection $0.827885$ and fault rejection $0.771780$ at $\zeta=1/2$.
At $M=2$ this estimator exceeds the $0.01$ normal-rejection limit; its fault
rejection is therefore shown as a diagnostic in the SM. At $M=512$ its saved
normal worst-case rejection is $0.049032$. The eight independent simulated
datasets give nominal pointwise intervals for linear metrics, while the
posterior theorem supplies the uniform MP comparison.
For the fixed 10\% fault traffic and $\zeta=1/2$, the saved branch data give
\begin{center}\small
\begin{tabular}{lccc}
\toprule
receiver and source & $A$ & $G$ & $G/A$\\
\midrule
Q, $M=2,w=1$ & 0.99583 & 0.95069 & 0.95467\\
Q, $M=2,w=0.95$ & 0.98772 & 0.93634 & 0.94799\\
MP, $M=2,w=1$ & 0.27841 & 0.08770 & 0.31499\\
MP, $M=512,w=1$ & 0.90828 & 0.86596 & 0.95341\\
Direct & 1.00000 & 0.95000 & 0.95000\\
Known-$P$ oracle & 0.95000 & 0.95000 & 1.00000\\
\bottomrule
\end{tabular}
\end{center}
Here $G$ is unnormalized clean-Bell overlap, while $G/A$ conditions on
acceptance. At $w=0.95$, Q has $G/A=0.94799$ against $0.95000$ for direct
delivery, alongside its greater fault rejection under the stated normal
protection specification. The oracle uses the source projector $P$.
For $w=1$, code-space logical $X$ and $Z$ are accepted with probability
one and produce Bell states orthogonal to the intended message.

\section{Finite dimension-four certificates, binary payload, and hardware noise}\label{supp:sm:finite-certificates-hardware}
\subsection{Contaminated collective-MP posterior blocks in dimension four}
For a unit vector \(v\), define
\begin{align}
 \Omega_g^{(M)}(v)&=\mathbb E_{P\sim\mathrm{Haar}}\rho_P^{\otimes M}\braket{v|P|v}/2,\\
 \Omega_b^{(M)}(v)&=\mathbb E_{P\sim\mathrm{Haar}}\rho_P^{\otimes M}\braket{v|Q|v}/2.
\end{align}
Schur--Weyl duality and \(U(4)\downarrow U(1)\times U(3)\) branching give scalar posterior blocks \(g_{\lambda\mu}\) and \(b_{\lambda\mu}\). Writing \(a=w/2\), \(b=(1-w)/2\),
\begin{equation}
 t_\lambda=\frac{s_\lambda(a,a,b,b)}{\dim V_\lambda},\qquad B_{\lambda\mu}=\sum_iq_i(\lambda,\mu)t_{\lambda+e_i},
\end{equation}
where \(\lambda\vdash M\) has at most four rows and \(\lambda_i\ge\mu_i\ge\lambda_{i+1}\). Write \(l_i=\lambda_i+4-i\), \(m_j=\mu_j+3-j\). The fundamental Gelfand--Tsetlin weights \cite{GrinkoBurchardtOzols2023} are
\begin{equation}
 q_i(\lambda,\mu)=\frac{\prod_{j=1}^3(l_i-m_j)}{\prod_{k\ne i}(l_i-l_k)},
 \qquad\lambda+e_i\text{ a partition}.
\end{equation}
They are nonnegative and sum to one. Each posterior block has dimension \(f_\lambda\dim V_\mu^{U(3)}\), including the full permutation multiplicity \(f_\lambda\). All dimensions sum to \(4^M\). Using \(\langle v|\rho_P|v\rangle=b+(a-b)\langle v|P|v\rangle\) yields
\begin{equation}
 g_{\lambda\mu}=\frac{B_{\lambda\mu}-bt_\lambda}{2(a-b)},\qquad b_{\lambda\mu}=\frac{t_\lambda}{2}-g_{\lambda\mu}.
\end{equation}
At \(w=1/2\), both are \(t_\lambda/4\). Their common support gives
\begin{equation}
 \Omega_b^{(M)}(v)\preceq\kappa_M(w)\Omega_g^{(M)}(v),\qquad \kappa_M=\max_{g_{\lambda\mu}>0}b_{\lambda\mu}/g_{\lambda\mu}.
\end{equation}

The support statement includes \(w=1\): for a vector \(x\), the function \(f_x(P)=\langle x|\rho_P^{\otimes M}|x\rangle\) is nonnegative. Both \(P_{vv}\) and \(Q_{vv}\) are positive Haar-almost everywhere. Thus their weighted integrals vanish exactly when \(\int f_x=0\), proving \(\ker\Omega_g=\ker\Omega_b=\ker T_M\). Every ratio is evaluated on this common support.

Let \(\{E_y\}\) be any joint reference POVM and let the rejection effect of the conditional data instrument be \(R_y=\sum_jr_{yj}\ket{v_{yj}}\bra{v_{yj}}\), \(0\leq R_y\leq I\). Then
\begin{align}
 \bar\alpha&=\sum_{yj}r_{yj}\Tr[E_y\Omega_g^{(M)}(v_{yj})],\\
 \bar D&=\sum_{yj}r_{yj}\Tr[E_y\Omega_b^{(M)}(v_{yj})]\leq\kappa_M\bar\alpha.
\end{align}
The same proof applies to continuous POVMs by integration. It admits independent quantum ancillas and arbitrary accepted quantum data output, and includes reference-preprocessing failures as outcomes. For two references,
\begin{equation}
 \kappa_2(w)=\frac{-5w^2+6w+2}{-5w^2+4w+3}.
\end{equation}

\subsection{Two-reference blocks and the sharp ROC slope}
For \(v=|0\rangle\), the stabilizer \(U(1)\times U(3)\) and reference swap give five sectors. Their dimensions and eigenvalues are
\begin{center}\small
\begin{tabular}{ccll}
\toprule
sector & dim. & \(240g\) & \(240b\)\\
\midrule
\(vv\) &1&\(3(w^2+1)\)&\(3(w^2-2w+2)\)\\
one \(v,+\)&3&\(3w^2-2w+4\)&\(3w^2-4w+5\)\\
one \(v,-\)&3&\(-5w^2+6w+2\)&\(-5w^2+4w+3\)\\
no \(v,+\)&6&\(3w^2-4w+5\)&\(3w^2-2w+4\)\\
no \(v,-\)&3&\(-5w^2+4w+3\)&\(-5w^2+6w+2\)\\
\bottomrule
\end{tabular}
\end{center}
An independent derivation expands \(\rho_P=(1-w)I/2+(2w-1)P/2\) and evaluates moments through degree three:
\begin{equation}
 \mathbb E\prod_a P_{i_a j_a}
 =\sum_{\sigma,\tau\in S_n}
 \prod_a\delta_{i_a,j_{\sigma(a)}}
 \mathrm{Wg}_4(\sigma^{-1}\tau)\,2^{c(\tau)}.
\end{equation}
Here \(\mathrm{Wg}_4\) is the inverse permutation Gram matrix with entries \(4^{c(\sigma^{-1}\tau)}\). This reconstructs the full 16-dimensional operators, including zero off-block entries.

Let \(L=3+4w-5w^2>0\). In table order the slacks \(\kappa_2 g-b\) are
\begin{align}
&\frac{(2w-1)(1+w-w^2)}{20L},\
 \frac{(2w-1)(7+2w-2w^2)}{240L},\nonumber\\
&\frac{(2w-1)(1+2w-2w^2)}{48L},\
 \frac{(2w-1)^3}{120L},\ 0.
\end{align}
All are nonnegative for \(1/2\le w\le1\); the final block attains the maximum. At \(w=0.9\), the slope is \(67/51\). An attainable local segment uses reference effect \(s\Pi_{v^\perp,-}\), \(0\le s\le1\), and its positive complement. On that outcome only, the data rejection/acceptance Kraus operators are \(|v\rangle\langle v|\) and \(I-|v\rangle\langle v|\); the other reference outcome accepts identically. It has \(\alpha=51s/1600,D=67s/1600\). At \(s=76/153\) it reaches the MP bound at the binary detector's false alarm. Haar randomization of the chosen known \(v\) makes the scores independent of \(P\). This proves the ROC segment, without optimizing its output disturbance.

\subsection{Coherent filter and output channel}
The antisymmetrizer on two references and data is \(\Pi_A=\frac16\sum_{\pi\in S_3}\operatorname{sgn}(\pi)V_\pi\). For filter strength \(0\le h\le1\), the complete binary instrument is
\begin{equation}
 K_a=I-h\Pi_A,\qquad K_r=\sqrt{2h-h^2}\Pi_A.
\end{equation}
It yields \(\alpha=(2h-h^2)(1-w^2)/12\) and \(D=(2h-h^2)(2w-w^2)/12\). For normal logical operators,
\begin{equation}
 \mathcal E_g(X)=c_{\rm id}X+c_P\Tr(X)P/2+c_Q\Tr(X)Q/2,
\end{equation}
where
\begin{align}
 c_{\rm id}&=1-h(1-w^2)/6+h^2(1-w)^2/36,\\
 c_P&=h^2w(1-w)/9,\qquad c_Q=h^2(1-w)/18.
\end{align}
Its nonnegative Weyl/replacement decomposition proves
\begin{equation}
 \delta_g=\frac{3c_P/4+c_Q}{c_{\rm id}+c_P+c_Q}=\frac{h^2(1-w)(3w+2)}{36(1-\alpha)}.
\end{equation}
To justify the reference-extended norm equality, decompose the \(P/2\) replacement into the four logical Weyl unitaries. The identity weight is \(G_g=c_{\rm id}+c_P/4\), the other three weights are \(c_P/4\), and the \(Q/2\) replacement weight is \(c_Q\). Convexity bounds the half-diamond error by \((3c_P/4+c_Q)/A_g\). On a logical Bell input, all nonidentity Weyl outputs and the complementary replacement are orthogonal to the target; hence this bound is attained. The nonnegative channel decomposition and the saturating Bell input establish the reference-extended norm equality.

Direct Kraus contraction determines the map on every matrix unit. In the canonical basis put \(t_i=w/2\) for \(i<2\), and \(t_i=(1-w)/2\) otherwise. Expand \(K=\sum_\pi c_\pi V_\pi\), apply it to \(|a,b,i\rangle\langle a,b,j|\), weight by \(t_at_b\), and trace the first two registers. Keeping precisely those terms whose two traced indices match gives \(c_{\rm id},c_P,c_Q\) above. Linearity extends this calculation to an untouched remote reference.

The ROC separation of the binary family is
\begin{equation}
 D-\kappa_2\alpha
 =\frac{(2h-h^2)(2w-1)(1+2w-2w^2)}
 {6(3+4w-5w^2)}>0
\end{equation}
for \(w>1/2\) and \(0<h\le1\). At \(w=1/2\), \(D=\alpha\); at \(w=1\), the normal output is exactly preserved.

\subsection{Exact-certificate definition and full objective}
An exact finite certificate consists of rational matrix-generation rules and rational coefficients that establish (i) a feasible accepted Choi matrix $J$ and positive reject completion, (ii) its scores and complete normal map, or (iii) a positive dual inequality for every accepted Choi matrix. A numerical solver output is a candidate for such a certificate, not its positivity proof.

Use the input-first, unnormalized Choi convention on $B^MDS$, where $B^M$
carries the $M$ reference inputs, $D$ is the message input, and $S$ is the
retained message output:
\begin{equation}
 J=\sum_{i,j}|i\rangle\langle j|_{B^MD}\otimes\mathcal I_a(|i\rangle\langle j|)_S.
\end{equation}
Thus $\mathcal I_a(X)=\Tr_{B^MD}[(X^T\otimes I_S)J]$, $J\succeq0$, and $\Tr_SJ\preceq I_{B^MD}$. If $Z=I-\Tr_SJ\succeq0$, the physical reject effect is $Z^T$ in this convention; $X\mapsto\Tr(Z^TX)|{\rm erase}\rangle\langle{\rm erase}|$ completes a trace-preserving instrument.

At the canonical $P_0=|0\rangle\langle0|+|1\rangle\langle1|$, let $\tau=\rho_{P_0}(9/10)^{\otimes M}$. Define
\begin{align}
 C_{A_g}^0&=(\tau\otimes P_0/2)^T\otimes I_S,\\
 C_{A_b}^0&=(\tau\otimes Q_0/2)^T\otimes I_S,\\
 C_{G_g}^0&=\tfrac14\tau^T\otimes|\Omega_2\rangle\langle\Omega_2|_{DS},
 \quad |\Omega_2\rangle=|00\rangle+|11\rangle.
\end{align}
Write $\overline C$ for the Haar twirl under $\overline U^{\otimes(M+1)}\otimes U$. Then $A_g=\Tr(\overline C_{A_g}J)$, $A_b=\Tr(\overline C_{A_b}J)$, $G_g=\Tr(\overline C_{G_g}J)$, and $D=1-A_b$. $G_g$ is the unnormalized logical Bell overlap, accounting for the factor $1/4$.

The complete covariant ROC optimization is
\begin{align}
 \text{maximize}\quad &1-\Tr(\overline C_{A_b}J)\nonumber\\
 \text{subject to}\quad &J\succeq0,\quad \Tr_SJ\preceq I,\nonumber\\
 &\Tr(\overline C_{A_g}J)\ge1-\alpha_0,\label{supp:eq:fullobjective}\\
 &\Tr[\{\overline C_{G_g}-(1-\delta_0)\overline C_{A_g}\}J]\ge0,\nonumber
\end{align}
with $J$ covariant and reference-permutation invariant. The stabilizer gives the normal map $aX+b\Tr(X)P/2+c\Tr(X)Q/2$. Complete positivity makes the logical Pauli weights $a+b/4,b/4,b/4,b/4$ and the orthogonal replacement weight $c$ nonnegative. Accordingly $\delta_g=1-G_g/A_g$, so the last constraint is the full reference-extended normal criterion in this covariant problem. It is not an identity for arbitrary channels inferred from Bell fidelity alone.

The proof of the fixed-budget theorem has three parts. A rational $J_3$ below attains the task for all $P$. A positive full-space M2 dual bounds every quantum instrument, even without the disturbance constraint. The posterior MP19 inequality bounds every reference-first MP instrument, also without disturbance. Padding by ignored references excludes all smaller maximum budgets. The first two are independently checkable from the rational data; the third is reconstructed by exact branching, not by a $4^{19}$ matrix.

\subsection{Exact three-copy resource certificate}
At \(w=9/10\), an exact rational three-reference instrument has \(\alpha=0.01989210100597\), \(\delta_g=0.01955045429812\), and \(D=0.11757171020084\). Its exact M2 dual bounds \(D\le4501613865629/50000000000000<0.095\) under the specified normal constraints. Together with the M19 MP inequality below, it establishes \(M_Q^*=3\), \(M_{MP}^*\ge20\) for \(\alpha,\delta_g\le0.02\), \(D\ge0.10\).

\subsection{Full-space rational representation and repair}
The certificate specifies rational coefficients of
\begin{equation}
 H_\pi=\frac{1}{2|S_M|}\sum_{s\in S_M}
 \left[\operatorname{PT}_{B^MD}V_{s\pi s^{-1}}
 +\operatorname{PT}_{B^MD}V_{s\pi^{-1}s^{-1}}\right].
\end{equation}
These are explicit rational matrices in the computational basis, commuting with \(\overline U^{\otimes(M+1)}\otimes U\) and the reference permutations. The 22 real M3 generators retain every mixed representation sector and multiplicity. Rounded coefficients of denominator \(10^{12}\) define \(J_{\rm rat}\). A single fixed repair is
\begin{align}
 J_3=\frac{9997}{10000}\Big[&
 \left(1-\frac1{50}-\frac1{10000}\right)J_{\rm rat}\nonumber\\
 &+\frac1{50}J_{\rm id}+\frac1{10000}J_{\rm dep}\Big].
\end{align}
\(J_{\rm id}\) traces references and returns data; \(J_{\rm dep}\) traces the input and returns \(I/4\). The identity admixture supplies quality and acceptance slack, the depolarizing admixture supplies CP interior, and the overall scale supplies TNI slack.

The sparse full \(1024\times1024\) Choi matrix has 80 torus-weight blocks. The \(256\times256\) input complement has 35. All off-block elements vanish exactly. For every distinct rational block \(A\), a floating Cholesky proposes an invertible rational triangular matrix \(R\), rounded to denominator \(10^8\). Exact arithmetic proves that \(RAR^T\) has positive diagonal and strict row diagonal dominance. Gershgorin and invertible congruence prove positivity without using a floating eigenvalue sign.

All four complex logical matrix units and all output entries are checked exactly to have the Weyl/replacement form. Its coefficients are
\begin{align}
 a&=766409138354447462258527/(8\,10^{23}),\\
 b&=4695874183145070714789/(4\,10^{23}),\\
 c&=4142716237243187997237/(4\,10^{23}).
\end{align}
The corresponding \(A_g=a+b+c\), \(G_g=a+b/4\), and \(\delta_g=(3b/4+c)/A_g\) meet the stated dimension-four task. Covariance transfers it to all \(P\). The positive reject effect can be completed with an erasure output. Every branch consumes three references, so maximum and expected counts coincide.

\subsection{Unrestricted M2 dual}
Rebuild \(\overline C_{A_g},\overline C_{A_b}\) by exact rational projection onto the complete commutant and verify residual orthogonality to every permutation generator. The certificate supplies \(Y\succeq0,\lambda\ge0\) such that
\begin{equation}
 \overline C_{A_b}+Y\otimes I-\lambda\overline C_{A_g}\succeq0.
\end{equation}
For every full accepted \(J\succeq0,\operatorname{Tr}_S J\preceq I\), this gives \(A_b\ge\lambda A_g-\operatorname{Tr}Y\). With \(A_g\ge49/50\) the stated strict M2 detection bound follows. No quality multiplier is needed. The rational reconstruction is tested in the complete input--output space, not only in reduced coordinates.

\subsection{Exact positivity and posterior checks}
\begingroup\raggedright
Exact positivity follows from rational congruence and strict diagonal dominance; the posterior tables reproduce the multiplicities and both traces $1/4$.
\par\endgroup

\paragraph{Measurement-first budget.}
For nineteen contaminated references the exact maximum is
\[
\kappa_{19}(9/10)=\frac{1513309885517371967}{304284815783363047}<5.
\]
Hence $D<0.10$ whenever $\alpha\le0.02$, excluding all maximum budgets through 19. This is distinct from the $d=8,w=0.95$ task.
\subsection{Equal accepted effects can have unequal payloads}
At \(h=1\), let \(S,T,A\) denote the three Young projectors. The coherent accepted Kraus operator is \(S+T\). Measuring and forgetting the full label instead produces \(X\mapsto SXS+TXT\) before references are discarded. Both have accepted effect \(S+T\). At \(w=1\), coherent binary processing has \(A_g=G_g=1\), while full-label readout gives \(A_g=1\), \(G_g=2/3\). Direct contraction of the six permutation terms gives on normal logical operators
\begin{equation}
 \mathcal A_{\rm label}(X)=\frac59X+\frac49\Tr(X)P/2\qquad(w=1).
\end{equation}
Thus $A_g=5/9+4/9=1$ and $G_g=5/9+(4/9)/4=2/3$. This verifies the channel distinction on every logical matrix unit.

\subsection{Native binary dilation and a bounded noise statement}
Let $S=S_{12}$ and $C=S_{12}S_{23}$. The permutation identity
\begin{equation}
 \Pi_A=(I-S)(I+C+C^2)/6
\end{equation}
allows coherent swap-parity and cycle-phase extraction. The cycle-invariant subspace carries the symmetric and antisymmetric one-dimensional representations; the swap-odd predicate selects the latter. The other Young labels are workspace, uncomputed after copying only this predicate to the flag. References and workspace are traced at the end, while the data are retained. The gate basis is all-to-all $R_z,R_y,\mathrm{CX}$, with final flag measurement and no intermediate measurement/reset. The $M=2,d=4,h=1$ binary dilation constructed for this study uses ten receiver qubits, 208 CX, and 348 one-qubit native gates, with depth 257 and CX depth 123 under its declared all-to-all ideal basis. Independent checks at \(h=0,0.35,1\) return work registers to zero. Three work qubits supplement the six reference/data qubits and one flag. The structured circuit reduces the gate count relative to an earlier direct construction by trading additional workspace for fewer gates. General Schur-transform circuits provide implementation background \cite{BaconChuangHarrow2006}. The count above belongs to this $M=2,d=4$ permutation/phase-extraction circuit, with the accepted channel fixed by compute/uncompute. It is not a gate count for the seven- or 399-reference constructions. The correction to Krovi's general Schur algorithm does not affect either this direct permutation circuit or the abstract projection theorem.

For local Markov Pauli channels after every gate, let the total nonidentity probabilities be \(p_2=10^{-5}\) for CX and \(p_1=10^{-6}\) for one-qubit gates. With preparation, idle, and other errors set to zero, each channel has half-diamond error exactly its total nonidentity probability. Telescoping gives \(\mu_{\rm inst}\le208p_2+348p_1=0.002428\). Define the Haar-averaged MP witness
\begin{equation}
 W=\overline D-\kappa_2(0.9)\overline\alpha,
 \qquad \overline D=\mathbb E_P D(P),\quad
 \overline\alpha=\mathbb E_P\alpha(P),
 \label{supp:eq:mp-witness}
\end{equation}
where \(\alpha(P)\) is the false-alarm probability on the normal \(P/2\)
input. The posterior inequality gives \(W\le0\) for every MP receiver with at
most two references, including arbitrary collective reference measurements. At \(w=0.9,h=1\), the ideal values are
\(\alpha=(1-w^2)/12=19/1200\) and \(D=(2w-w^2)/12=33/400\). Thus the
uniform bounds \(D(P)\ge D_{\rm low}=33/400-\mu\) and
\(\alpha_{\rm wc}(P)\le\alpha_{\rm up}=19/1200+\mu\), together with
\(\kappa_2(0.9)=67/51\) and \(\mu=0.002428\), imply
\begin{equation}
 W\ge D_{\rm low}-\frac{67}{51}\alpha_{\rm up}
 =\frac{1072561}{19125000}=0.056081620915\ldots.
 \label{supp:eq:mp-witness-noise}
\end{equation}
In particular, \(\alpha_{\rm wc}\le0.01826134\), \(D\ge0.080072\), and
\(\delta_g\le0.01821193\). These bounds apply to the stated local Markov
Pauli model with preparation, idle, and other errors set to zero. The
gate-level bounds concern this $d=4$ circuit. If the complete ideal/noisy
instruments differ by half-diamond distance at most \(\mu\), their accepted
subnormalized states differ by trace norm at most \(2\mu\), and acceptance
probabilities differ by at most \(\mu\). For actual acceptance \(A\ge A_0-\mu\),
\begin{equation}
 \frac12\left\|\frac{\sigma_a}{A}-\frac{\tau_a}{A_0}\right\|_1
 \le\frac{3\mu}{2A}\le\frac{2\mu}{A_0-\mu}.
 \label{supp:eq:normalization-transfer}
\end{equation}
Adding the ideal error gives \(\delta_g\le\delta_0+2\mu/(A_0-\mu)\), uniformly over \(P,R\) and normal inputs. This bounds conditional normalized states, not the diamond norm of an input-dependent normalization.

\subsection{Noise channels, tests, and failure regions}
On a local Hilbert space of dimension \(q\), the simulated Pauli channel is
\begin{equation}
 \mathcal N_p=(1-p)\mathrm{id}+
 \frac{p}{q^2-1}\sum_{P\ne I}P(\cdot)P^\dagger.
\end{equation}
The nonidentity Pauli outputs on a local Bell state are orthogonal to the identity output. This gives half-diamond distance \(p\), with the matching upper bound from convexity. Its fully depolarizing mixing weight is \(pq^2/(q^2-1)\). Coherent errors are \(e^{-i\theta XX/2}\), \(e^{-i\theta Y/2}\), and \(e^{-i\theta Z/2}\) after CX, \(R_y\), and \(R_z\), respectively. A final classical flag flip has probability \(q_{\rm ro}\). Each axis is varied separately. Preparation and idle errors are zero in this model.

Sixteen calibration encodings and 64 independently generated evaluation encodings are drawn by complex-Gaussian QR, with seeds 94001--94016 and 95001--95064. The actual gate channels give both logical Choi branches, including the remote reference. Work and reference outputs are traced, including noisy nonzero work outcomes.

Let \(J_a\) be the input-first unnormalized Choi matrix and \(H=(\operatorname{Tr}_{\rm out}J_a)^T\). Define \(\bar A=\operatorname{Tr}H/2\), \(A_{\min}=\lambda_{\min}(H)\), and \(\Delta J=J_a-\bar A J_{\rm id}\). A feasible diamond-norm dual gives
\begin{equation}
 u=\tfrac12\|\operatorname{Tr}_{\rm out}|\Delta J|\|_\infty,\qquad
 \delta_g\le\frac{u+\tfrac12\|H-\bar A I\|_\infty}{A_{\min}}.
\end{equation}
The numerical calculation pads \(u\) by \(2\times10^{-12}\). Bell-input half-trace distance is \(\|\Delta J\|_1/(4\bar A)\), a lower bound only. A numerical pass requires the displayed upper bound at most 0.02, \(1-A_{\min}\le0.02\), and \(D\ge0.08\); a Bell lower bound above 0.02 establishes failure. Cases between the two bounds are classified as unresolved.

\begin{center}\small
\begin{tabular}{lrrr}
\toprule
condition & pass & fail & unresolved\\
\midrule
ideal &64&0&0\\
\(p_2=10^{-5}\) &64&0&0\\
\(p_2=3\,10^{-5}\) &64&0&0\\
\(p_2=10^{-4}\) &0&64&0\\
\(p_2=3\,10^{-4}\) &0&64&0\\
\(\theta=-0.001,+0.001\) (each)&64&0&0\\
\(\theta=-0.003\)&16&37&11\\
\(\theta=+0.003\)&15&37&12\\
\(q_{\rm ro}=0.001,0.003\) (each)&64&0&0\\
\bottomrule
\end{tabular}
\end{center}
Telescoping proves the all-$P$ bounds at $p_2=10^{-5}$ and separately at $q_{\rm ro}=0.001$. The remaining table entries classify the finite evaluation ensemble under the stated noise model. The application pass criterion jointly checks normal rejection, conditional disturbance, and detection.

\subsection{Parameter comparison across retained scenarios}
Here $\epsilon$ bounds both normal rejection and conditional disturbance, and $D_0$ is the required detection probability. The source and target columns define distinct tasks.
\begin{center}\small
\begin{tabular}{lcccccl}
\toprule
scenario & $d$ & $w$ & $\epsilon$ & $D_0$ & quantum budget & MP necessary budget\\
\midrule
ideal source & $\ge3$ & $1$ & $0.01$ & $0.50$ & $M_{\rm Q}^*\le7$ & $M_{\rm MP}^*\ge98(d-2)$\\
contaminated source & $8$ & $0.95$ & $0.01$ & $0.08$ & $M_{\rm Q}^*\le2$ & $M_{\rm MP}^*\ge84$\\
exact finite certificate & $4$ & $0.90$ & $0.02$ & $0.10$ & $M_{\rm Q}^*=3$ & $M_{\rm MP}^*\ge20$\\
local gate-noise model & $4$ & $0.90$ & $0.02$ & $0.08$ & $M=2$ & ---\\
\bottomrule
\end{tabular}
\end{center}
The first two quantum budgets are feasible constructions; the third is an exact minimum for its task. The MP entries are necessary lower bounds. The final row specifies the circuit model above with $p_2=10^{-5}$, $p_1=10^{-6}$; its proven detection bound is $0.080072$ and differs from the certificate's $0.10$ target.

\section{Relation to primary literature}\label{supp:sm:literature}
Quantum references can serve as a program for an operation on an independent data register, as in programmable quantum processors~\cite{NielsenChuang1997}. Learning quantum measurements~\cite{BisioDAranoPerinottiSedlak2011} and rotations about an unknown direction~\cite{MoChiribella2019} provides related settings in which retaining quantum information benefits a later operation. Programmable discrimination and universal learning machines use finite quantum training sets to classify a new system~\cite{DusekBuzek2002,BergouHillery2005,Sentis2010,FanizzaMariGiovannetti2019}. The role of memory depends on the objective: for minimum-error classification of two unknown qubit states, an optimized training-set measurement and a record-controlled data measurement can attain the coherent programmable optimum~\cite{Sentis2012}. Here the objective combines a one-sided leakage decision with exact preservation of every normal message and its remote entanglement.

The row-cutoff predicate used in Sec.~\ref{supp:sm:ideal-arbitrary-m} is the one-sided rank-testing predicate of O'Donnell and Wright~\cite{ODonnellWright2017} (Sec.~6.1, Proposition~6.1 in the \href{https://www.cs.cmu.edu/~jswright/papers/quantum-spectrum-testing.pdf}{author-posted preprint}).  Their support argument supplies the normal-support step, including all multiplicities, while the present calculation evaluates the mixed-reference single-message instrument and its finite-$M$ law.  Universal comparison of unknown pure states and many unknown systems supplies earlier antisymmetric-comparison antecedents~\cite{BarnettCheflesJex2003,JexAnderssonChefles2003}.  The task here compares a normal query with an independently leaking query under a common unknown subspace, rather than comparing copies of one state.

Quantum-memory separations for state learning supply broader sample-complexity context. In the without-memory state-learning model of Chen, Cotler, Huang, and Li~\cite{ChenCotlerHuangLi2021}, each access returns the classical outcome of an arbitrary POVM on one sample, chosen adaptively from earlier outcomes (Definitions~4.15--4.18 of arXiv:2111.05881v2). Bubeck, Chen, and Li~\cite{BubeckChenLi2020} consider unentangled, single-copy measurement schedules (Definition~3.2 of arXiv:2004.07869). Our MP class allows an arbitrary collective POVM on all references before the message operation. These earlier restrictions therefore do not establish its converse; the posterior bound here does not depend on separate reference measurements.

Li et al.~\cite{CoherentInference2026} already allow independent quantum data--reference inputs in density-matrix exponentiation and a comparator with arbitrary collective reference POVMs followed by conditional data channels (Sec.~S5 and Lemma~S5.1 of arXiv:2605.21457v1). Independent quantum output and this access order are thus shared features. Their DME error measures diamond-norm distance from the target unitary channel. Here zero false alarm and zero normal disturbance constrain the normal inputs, while anomalous messages can still be missed. These are different error objectives. The present result gives the full zero-false-alarm performance for the nonidentically distributed reference--message family, protected normal quantum output, and the task-specific exact MP posterior slope and matching sample order.

Known-code subspace verification assumes a classically specified target projector or stabilizer structure~\cite{ChenSubspace2025}. Spectral anomaly detection studies spectral scores under Hamiltonian, block-encoding, or sample-based DME access~\cite{SpectralAnomaly2026}. These access models permit coherent use of the spectral information and do not require all references to be converted to a classical record before the message operation. The MP posterior bound applies to that reference-first classicalization restriction; its decision constraint is independent of the quality of the returned message.

Yang, Chiribella, and Ebler~\cite{YangChiribellaEbler2016} give zero-error compression of identically prepared mixed-state ensembles in general dimension and rank (Appendix~F, Theorem~4 of arXiv:1506.03542). The Schur decomposition separates a maximally mixed multiplicity factor, which can be discarded and restored by the decoder. Measuring a reference-only Schur label preserves the reference ensemble, which is already block diagonal in that label. The one-versus-$2/3$ normal Bell-fidelity example in Sec.~\ref{supp:sm:finite-certificates-hardware} concerns measurement of the full Young label of the references and message jointly. The remaining memory question is the minimum quantum information required for the later support check, allowing a classical record and a specified normal-output tolerance, rather than faithful reconstruction of the entire reference ensemble.

\end{document}